\documentclass[letterpaper,11pt]{article}
\usepackage{times}
\usepackage[left=1in, right=1in, top=1in, bottom=1in,letterpaper]{geometry}
\usepackage{amsmath,amsthm,amssymb,paralist}
\usepackage{graphicx}
\usepackage[dvipsnames]{xcolor}
\usepackage{color}
\usepackage{wrapfig}
\usepackage[small,bf]{caption}
\usepackage{hyperref}
\usepackage{paralist}
\usepackage[small,compact]{title sec}
\usepackage{listings}
\usepackage{xcolor}
\usepackage{textcomp}
\usepackage{placeins}
\usepackage{bbm}
\usepackage[ruled,noline,noend,linesnumbered]{algorithm2e}
\usepackage[normalem]{ulem}
\usepackage{newfloat}
\DeclareFloatingEnvironment[name={Supplementary Figure}]{suppfigure}
\usepackage{caption}
\newcommand{\fabio}{\textcolor{red}}

\newcommand{\diego}{\textcolor{orange}}

\usepackage{float}
\usepackage[caption = false]{subfig}

\usepackage{tikz}
\usetikzlibrary{positioning}

\let\OLDthebibliography\thebibliography
\renewcommand\thebibliography[1]{
  \OLDthebibliography{#1}
  \setlength{\parskip}{0pt}
  \setlength{\itemsep}{0pt plus 0.3ex}
}

\newcommand{\algname}{$\mathsf{SPRISS}$}
\newcommand{\algnameshort}{$\mathsf{SP}$}
\newcommand{\sakeima}{\textsc{SAKEIMA}}
\newcommand{\sakeimashort}{$\mathsf{SK}$}
\newcommand{\exactshort}{$\mathsf{E}$}
\newcommand{\kmc}{\textsc{KMC}}
\newcommand{\jellyfish}{\textsc{Jellyfish}}
\newcommand{\E}{\mathbb{E}}

\newcommand{\D}{\mathcal{D}}

\newtheorem{proposition}{Proposition}
\newtheorem{definition}{Definition}
\newtheorem{lemma}{Lemma}

\newenvironment{packed_itemize}{
\begin{itemize}
	\setlength{\itemsep}{1pt}
	\setlength{\parskip}{0pt}
	\setlength{\parsep}{0pt}
}{\end{itemize}}

\usepackage{enumitem}
\newlist{todolist}{itemize}{2}
\setlist[todolist]{label=$\square$}
\usepackage{pifont}

\begin{document}
\title{\algname: Approximating Frequent $k$-mers by Sampling Reads, \\
 and Applications\thanks{Part of this work was supported by the MIUR, the Italian Ministry of Education, University and Research, under PRIN Project n. 20174LF3T8 AHeAD (Efficient Algorithms for HArnessing Networked Data) and the initiative ``Departments of Excellence" (Law 232/2016), and by the Univ. of Padova under project SEED 2020 RATED-X.}
}
\author{Diego Santoro\thanks{Department of Information Engineering, University of Padova, Padova (Italy).}
\footnotemark[3]
\\\texttt{diego.santoro@dei.unipd.it}
\and
Leonardo Pellegrina\footnotemark[2]
\thanks{These authors contributed equally to this work.}
\\\texttt{pellegri@dei.unipd.it}
\and
Fabio Vandin\footnotemark[2]
\thanks{Corresponding author.}
\\\texttt{fabio.vandin@unipd.it}}

\date{}

\maketitle
\thispagestyle{empty}

\begin{abstract}
The extraction of $k$-mers is a fundamental component in many complex analyses of large next-generation sequencing datasets, including reads classification in genomics and the characterization of RNA-seq datasets. The extraction of all $k$-mers and their frequencies is extremely demanding in terms of running time and memory, owing to the size of the data and to the exponential number of $k$-mers to be considered. However, in several applications, only \emph{frequent} $k$-mers, which are $k$-mers appearing in a relatively high proportion of the data, are required by the analysis. In this work we present \algname, a new efficient algorithm to approximate frequent $k$-mers and their frequencies in next-generation sequencing data. \algname\ employs a simple yet powerful reads sampling scheme, which allows to extract a representative subset of the dataset that can be used, in combination with any $k$-mer counting algorithm, to perform downstream analyses in a fraction of the time required by the analysis of the whole data, while obtaining comparable answers. Our extensive experimental evaluation demonstrates the efficiency and accuracy of \algname\ in approximating frequent $k$-mers, and shows that it can be used in various scenarios, such as the comparison of metagenomic datasets and the identification of discriminative $k$-mers, to extract insights in a fraction of the time required by the analysis of the whole dataset.
\end{abstract}

\vspace{1cm}
\textbf{keywords}: $k$-mer analysis; frequent $k$-mers; read sampling; pseudodimension;

%
%

\newpage

\setcounter{page}{1}

\section{Introduction}

The study of substrings of length $k$, or $k$-mers, is a fundamental task in the analysis of large next-generation sequencing datasets. The extraction of $k$-mers, and of the frequencies with which they appear in a dataset of reads, is a crucial step in several applications, including the comparison of datasets and reads classification in metagenomics~\cite{wood2014kraken}, the characterization of variation in RNA-seq data~\cite{audoux2017kupl}, the analysis of structural changes in genomes~\cite{li2003estimating,liu2017unbiased}, RNA-seq quantification~\cite{patro2014sailfish,zhang2014rna}, fast search-by-sequence over large high-throughput sequencing repositories~\cite{solomon2016fast}, genome comparison~\cite{sims2009alignment}, and error correction for genome assembly~\cite{kelley2010quake,salmela2016accurate}. 

$k$-mers and their frequencies can be obtained with a linear scan of a dataset. However, due to the massive size of the modern datasets and the exponential growth of the $k$-mers number (with respect to $k$), the extraction of $k$-mers is an extremely computationally intensive task, both in terms of running time and memory~\cite{elworth2020petabytes}, and several algorithms have been proposed to reduce the running time and memory requirements (see Section~\ref{sec:related_work}). Nonetheless, the extraction of all $k$-mers and their frequencies from a reads dataset is still highly demanding in terms of time and memory (e.g., \kmc~3~\cite{kokot2017kmc}, one of the currently best performing tools for $k$-mer counting, requires more than $2.5$~hours, $34$~GB of memory, and $500$~GB of space on disk on a sequence of $729$~Gbases~\cite{kokot2017kmc}, and from our experiments more than $30$ minutes, $300$~GB of memory, and $97$~GB of disk space for counting $k$-mers from \texttt{Mo17} dataset\footnote{Using $k=31$, $32$ workers, and maximum RAM of $350$ GB. See Supplemental Table~\ref{tab:dataMOB} for the size of \texttt{Mo17}.}).

While some applications, such as error correction~~\cite{kelley2010quake,salmela2016accurate} or reads classification~\cite{wood2014kraken}, require to identify \emph{all} $k$-mers, even the ones that appear only once or few times in a dataset, other analyses, such as the comparison of abundances in metagenomic datasets~\cite{benoit2016multiple,danovaro2017submarine,dickson2017carryover,sakeima} or  the discovery of $k$-mers discriminating between two datasets~\cite{ounit2015clark,liu2017unbiased}, hinge on the identification of \emph{frequent} $k$-mers, which are $k$-mers appearing with a (relatively) high frequency in a dataset. For the latter analyses, tools capable of efficiently extracting frequent $k$-mers only would be extremely beneficial and much more efficient than tools reporting all $k$-mers (given that a large fraction of $k$-mers appear with extremely low frequency). However, the efficient identification of frequent $k$-mers and their frequencies is still relatively unexplored (see Section~\ref{sec:related_work}).

A natural approach to speed-up the identification of frequent $k$-mers is to analyze only a \emph{sample} of the data, since frequent $k$-mers appear with high probability in a sample, while unfrequent $k$-mers appear with lower probability. A major challenge in sampling approaches is how to rigorously relate the results obtained analyzing the sample and the results that would be obtained analyzing the whole dataset. Tackling such challenge requires to identify a minimum sample size which guarantees that the results on the sample well represent the results to be obtained on the whole dataset. 
An additional challenge in the use of  sampling for the identification of frequent $k$-mers is due to the fact that, for values of $k$ of interest in modern applications (e.g., $k\in[20,60]$), even the most frequent $k$-mers appear in a relatively low portion of the data (e.g., $10^{-7}\text{-}10^{-5}$). The net effect is that the application of standard sampling techniques to rigorously approximate frequent $k$-mers results in sample sizes \emph{larger} than the initial dataset.

\subsection{Our Contributions}
\label{sec:our_contribs}
In this work we study the problem of approximating frequent $k$-mers in a dataset of reads. In this regard, our contributions are:
\begin{packed_itemize}
\item We propose \algname, \underline{S}am\underline{P}ling \underline{R}eads algor\underline{I}thm to e\underline{S}timate frequent $k$-mer\underline{S}\footnote{\url{https://vec.wikipedia.org/wiki/Spriss}}. \algname\ is based on a simple yet powerful read sampling approach, which renders \algname~very flexible and suitable to be used in combination with \emph{any} $k$-mer counter. In fact, the read sampling scheme of \algname\ returns a \emph{representative} subset of a dataset of reads, which can be used to obtain representative results for down-stream analyses based on frequent $k$-mers. 
\item We prove that \algname\ provides rigorous guarantees on the quality of the approximation of the frequent $k$-mers. In this regard, our main technical contribution is the derivation of the sample size required by \algname, obtained through the study of the pseudodimension~\cite{pollard1984convergence}, a key concept from statistical learning theory, of $k$-mers in reads.
\item We show on several real datasets that \algname\ approximates frequent $k$-mers with high accuracy, while requiring a fraction of the time needed by approaches that analyze all $k$-mers in a dataset.
\item We show the benefits of using the approximation of frequent $k$-mers obtained by \algname\ in two applications: the comparison of metagenomic datasets, and the extraction of discriminative $k$-mers. In both applications \algname\ significantly speeds up the analysis, while providing the same insights obtained by the analysis of the whole data.
\end{packed_itemize}

\subsection{Related Works}
\label{sec:related_work}

The problem of exactly counting $k$-mers in datasets has been extensively studied, with several methods proposed for its solution~\cite{kurtz2008new,marccais2011fast,melsted2011efficient,rizk2013dsk,audano2014kanalyze,roy2014turtle,kokot2017kmc,pandey2017squeakr}. Such methods are typically highly demanding in terms of time and memory when analyzing large high-throughput sequencing datasets~\cite{elworth2020petabytes}. For this reason, many methods have been recently developed to compute approximations of the $k$-mers abundances to reduce the computational cost of the task (e.g,~\cite{melsted2014kmerstream,sivadasan2016kmerlight,mohamadi2017ntcard,chikhi2013informed,zhang2014these,pandey2017squeakr}). However, such methods do not provide guarantees on the accuracy of their approximations that are simultaneously valid for all (or the most frequent) $k$-mers.
In recent years other problems closely related to the task of counting $k$-mers have been studied, including
how to efficiently index~\cite{pandey2018mantis,harris2020improved,marchet2020resource,marchet2020reindeer}, represent~\cite{chikhi2014representation,dadi2018dream,almodaresi2018space,guo2019degsm,guo2019degsm,marchet2019indexing,holley2020bifrost,rahman2020representation}, query~\cite{solomon2016fast,solomon2018improved,yu2018seqothello,sun2018allsome,bradley2019ultrafast,marchet2019data}, and store~\cite{hosseini2016survey,numanagic2016comparison,hernaez2019genomic,rahman2020disk} the massive collections of sequences or of $k$-mers that are extracted from the data.

A natural approach to reduce computational demands is to analyze a small sample instead of the entire dataset. 
To this end, methods that perform a downsampling of massive datasets have been recently proposed~\cite{brown2012reference,wedemeyer2017improved,coleman2019diversified}. These methods focus on discarding reads of the datasets that are very similar to the reads already included in the sample, computing approximate similarity measures as each read is considered. Such measures (i.e., the Jaccard similarity) are designed to maximise the diversity of the content of the reads in the sample. This approach is well suited for applications where rare $k$-mers are important, but they are less relevant for analyses, of interest to this work, where the most frequent $k$-mers carry the major part of the information. Furthermore, these methods have a heuristic nature, and do not provide guarantees on the relation between the accuracy of the analysis performed on the sample w.r.t. the analysis performed on the entire dataset. 
\sakeima~\cite{sakeima} is the first sampling method that provides an approximation of the set of frequent $k$-mers (together with their estimated frequencies) with rigorous guarantees, based on counting only a subset of all occurrences of $k$-mers, chosen at random.
\sakeima\ performs a full scan of the entire dataset, in a streaming fashion, and processes each $k$-mer occurence according to the outcome of its random choices. \algname, the algorithm we present in this work, is instead the first sampling algorithm to approximate frequent $k$-mers (and their frequencies), with rigorous guarantees, by sampling \emph{reads} from the dataset. In fact, \algname\ does not require to receive in input and to scan the entire dataset, but, instead, it needs in input only a small sample of reads drawn from the dataset, sample that may be obtained, for example, at the time of the physical creation of the whole dataset.
While the sampling strategy of \sakeima\ could be analyzed using the concept of \emph{VC dimension}~\cite{vapnik1998statistical}, the reads-sampling strategy of \algname\ requires the more sophisticated concept of \emph{pseudodimension}~\cite{pollard1984convergence}, for its analysis.

In this work we consider the use of \algname\ to speed up the computation of the Bray-Curtis distance between metagenomic datatasets and the identification of discriminative $k$-mers. Computational tools for these problems have been recently proposed~\cite{benoit2016multiple,saavedra2020mining}. These tools are based on exact $k$-mer counting strategies, and the approach we propose with \algname\ could be applied to such strategies as well.

\section{Preliminaries}
Let $\Sigma$ be an alphabet of $\sigma$ symbols. A dataset $\D = \{r_1, \dots, r_n\}$ is a bag of $|\D| = n$ \emph{reads}, where, for $i \in \{1,\dots,n\}$, a read $r_i$ is a string of length $n_i$  built from $\Sigma$.
For a given integer $k$, a $k$\emph{-mer} $K$ is a string of length $k$ on $\Sigma$, that is $K \in \Sigma^k$.  
Given a $k$-mer $K$, a read $r_i$ of $\D$, and a position $j \in \{0,\dots,n_i - k\}$, we define the indicator function $\phi_{r_i,K}(j)$ to be $1$ if $K$ \emph{appears} in $r_i$ at position $j$, that is $K[h] = r_i[j+h]$ $\forall h \in \{0,\dots,k-1\}$, while $\phi_{r_i,K}(j)$ is $0$ otherwise. 
The size $t_{\D,k}$ of the multiset of $k$-mers that appear in $\D$ is $t_{\D,k} = \sum_{r_i \in \D} (n_i -k +1)$. The average size of the multiset of $k$-mers that appear in a read of $\D$ is $\ell_{\D,k} = t_{\D,k}/n$, while the maximum value of such quantity is $\ell_{\max,\D,k} = \max_{r_i \in \D}(n_i - k + 1)$. The \emph{support} $o_\D(K)$  of $k$-mer $K$ in dataset $\D$ is the number of distinct positions of $\D$ where $k$-mer $K$ appears, that is $o_\D(K) = \sum_{r_i \in \D} \sum_{j = 0}^{n_i - k} \phi_{r_i,K}(j)$. The \emph{frequency} $f_\D(K)$ of a $k$-mer $K$ in $\D$ is the fraction of all positions in $\D$ where $K$ appears, that is $f_\D(K) = o_\D(K)/t_{\D,k}$.

The task of finding \emph{frequent $k$-mers} (FKs) is defined as follows: given a dataset $\D$, a positive integer $k$, and a \emph{minimum frequency threshold} $\theta \in (0,1]$, find the set $FK(\D,k,\theta)$ of all the $k$-mers whose frequency in $\D$ is at least $\theta$, and their frequencies, that is 
$FK(\D,k,\theta) = \{(K,f_\D(K)): K \in \Sigma^k, f_\D(K) \geq \theta \}$.

The set of frequent $k$-mers can be computed by scanning the dataset and counting the number of occurrences for each $k$-mers. However, when dealing with a massive dataset $\D$, the exact computation of the set $FK(\D,k,\theta)$ requires large amount of time and memory. For this reason, one could instead focus on finding an \emph{approximation} of $FK(\D,k,\theta)$ with rigorous guarantees on its quality. In this work we consider the following approximation, introduced in \cite{sakeima}.

\begin{definition}
	\label{def:approximation}
	Given a dataset $\D$, a positive integer $k$, a frequency threshold $\theta \in (0,1]$, and an accuracy  parameter $\varepsilon  \in  (0,\theta)$, an \emph{$\varepsilon$-approximation} $\mathcal{C}= \{(K,f_K) : K \in \Sigma^k, f_K \in [0,1]\}$ of $FK(\D,k,\theta)$ is a set of pairs $(K,f_K)$ with the following properties:
	\begin{packed_itemize}
		\item $\mathcal{C}$ contains a pair $(K,f_K)$ for every $(K, f_\D(K)) \in FK(\D,k,\theta)$;
		\item $\mathcal{C}$ contains no pair $(K,f_K)$ such that $f_\D(K) < \theta - \varepsilon$;
		\item for every $(K,f_K) \in \mathcal{C}$, it holds $|f_\D(K) - f_K| \leq \varepsilon/2$. 
	\end{packed_itemize}
\end{definition}
\noindent

Intuitively, the approximation $\mathcal{C}$ contains no \emph{false negatives} (i.e. all the frequent $k$-mers  in $FK(\D,k, \theta)$ are in $C$) and no $k$-mer whose frequency in $\D$
is much smaller than $\theta$. In addition, the frequencies in $\mathcal{C}$ are good approximations of the actual frequencies in $\D$, i.e. within a small error $\varepsilon/2$.

\begin{definition}
	Given a dataset $\D$ of $n$ reads, we define a \emph{reads sample} $S$ of $\D$ as a bag of $m$ reads, sampled independently and uniformly at random, with replacement, from the bag of reads in $\D$.
\end{definition}

A natural way to compute an approximation of the set of frequent $k$-mers is by processing a \emph{sample}, i.e. a small portion of the dataset $\D$, instead of the whole dataset. While previous work~\cite{sakeima} considered samples obtained by drawing $k$-mers independently from $\D$, we consider samples obtained by drawing entire \emph{reads}. As explained  in Section~\ref{sec:our_contribs}, our approach has several advantages, including the fact that it can be combined with any efficient $k$-mer counting procedure, and that it can be used to extract a \emph{representative} subset of the data on which to conduct down-stream analyses obtaining, in a fraction of the time required to process the whole dataset, the same insights. Such representative subsets could be stored and used for exploratory analyses, with a gain in terms of space and time requirements compared to using the whole dataset.

However, the development of an efficient scheme to effectively approximate the frequency of all frequent $k$-mers by sampling reads is highly nontrivial, due to dependencies among $k$-mers appearing in the same read. In the next sections, we develop and analyze algorithms to approximate $FK(\D,k, \theta)$ by read sampling, starting from a straightforward, but inefficient, approach (Section~\ref{sec:warmup}), then showing how pseudodimension can be used to improve the sample size required by such approach (Section~\ref{sec:pseudo_alg}), and culminating in our algorithm~\algname, the first efficient algorithm to approximate frequent $k$-mers by read sampling (Section~\ref{sec:pseudo_alg_bags}).

\section{Warm-Up: A Simple Algorithm for Approximating Frequent $k$-mers by Sampling Reads}
\label{sec:warmup}

A first, simple approach to approximate the set $FK(\D,k, \theta)$ of frequent $k$-mers consists in taking a sample $S$ of $m$ reads, with $m$ large enough, and report in output the set $FK(S,k,\theta-\varepsilon/2)$ of $k$-mers that appear with frequency at least $\theta-\varepsilon/2$ in the sample $S$. The following result, obtained by combining Hoeffding's inequality~\cite{mitzenmacher2017probability} and a union bound, provides an upper bound to the number $m$ of reads required to have guarantees on the quality of the approximation (see Supplement Material Section \ref{sec:A} for the full analysis).

\begin{proposition}
	\label{prop:samplebound1}
	Consider a sample $S$ of $m$ reads from $\D$. For fixed frequency threshold $\theta \in (0,1]$, error parameter $\varepsilon \in (0,\theta)$, and confidence parameter $\delta \in (0,1)$, if 
$		m \geq \frac{2}{\varepsilon^2} \left(\frac{\ell_{\max,\D,k}}{\ell_{\D,k}} \right)^2 \left( \ln \left( 2 \sigma^k\right)+\ln\left(\frac{1}{\delta} \right) \right)$
	then, with probability $\ge 1- \delta$, $FK(S,k,\theta-\varepsilon/2)$ is an $\varepsilon$-approximation of $FK(\D,k,\theta)$. 
\end{proposition}

While the result above provides a first bound to the number $m$ of reads required to obtain a rigorous approximation of the frequent $k$-mers,  it usually results in a sample size $m$ larger than $|\D|$ (this is due to the need for $\varepsilon$ to be small in order to obtain meaningful approximations, see Section~\ref{sec:exp1}), making the sampling approach useless. Thus, in the next sections we propose advanced methods to reduce the sample size $m$.

\section{A First Improvement: A Pseudodimension-based Algorithm for $k$-mers Approximation by Sampling Reads}
\label{sec:pseudo_alg}
In this section we introduce the notion of pseudodimension and we use it to improve the bound on the sample size $m$ of Proposition~\ref{prop:samplebound1}.

Let $\mathcal{F}$ be a class of real-valued functions from a domain $X$ to $[a,b] \subset \mathbb{R}$. Consider, for each $f \in \mathcal{F}$, the subset of $X' = X \times [a,b]$ defined as $R_f = \{ (x,t) : t \leq f(x) \}$, and call it \emph{range}. Let $\mathcal{F}^+ = \{R_f , f\in \mathcal{F}\}$ be a \emph{range set} on $X'$, and its corresponding \emph{range space} $Q'$ be $Q'=(X',\mathcal{F}^+)$. 
We say that a subset $D \subset X'$ is \emph{shattered} by $\mathcal{F}^+$ if the size of the \emph{projection set} $proj_{\mathcal{F}^+}(D) = \{ r \cap D: r \in \mathcal{F}^+ \}$ is equal to $2^{|D|}$. The \emph{VC dimension} $VC(Q')$ of $Q'$ is the maximum size of a subset of $X'$ shattered by $\mathcal{F}^+$.
The \emph{pseudodimension} $PD(X,\mathcal{F})$ is then defined as the VC dimension of $Q'$: $PD(X,\mathcal{F}) = VC(Q')$.

Let $\pi$ be the uniform distribution on $X$, and let $S$ be a sample of $X$ of size  $|S| = m$, with every element of $S$ sampled independently and uniformly at random from $X$. We define, $\forall f \in \mathcal{F}$, $f_S = \frac{1}{m} \sum_{x \in S}f(x)$ and $f_X =\E_{x\sim\pi}[f(x)]$. Note that $\E[f_S] = f_X$.
The following result relates the accuracy and confidence parameters $\varepsilon$,$\delta$ and the pseudodimension with the probability that the expected values of the functions in $\mathcal{F}$ are well approximated by their averages computed from a finite random sample.

\begin{proposition}[\cite{talagrand1994sharper,long1999complexity}]
	\label{prop:pseudosamplesize}
	Let $X$ be a domain and $\mathcal{F}$ be a class of real-valued functions from $X$ to $[a,b]$. Let $PD(X,\mathcal{F}) = VC(Q') \leq v$. There exist an absolute positive constant $c$ such that, for fixed $\varepsilon,\delta \in (0,1)$, if $S$ is a random sample of $m$ samples drawn independently and uniformly at random from $X$ with
$		m \geq \frac{c \left( b-a \right)^2}{\varepsilon^2}\left(v + \ln \left(\frac{1}{\delta}\right)\right)$ 
	then, with probability $\geq 1-\delta$, it holds simultaneously $\forall f \in \mathcal{F}$ that $| f_S  - f_X | \leq \varepsilon$.
\end{proposition}

The universal constant $c$ has been experimentally estimated to be at most $0.5$ \cite{loffler2009shape}. 


We now define the range space associated to $k$-mers, derive an upper bound to its pseudodimension, and use the result above to derive an improved bound on the number $m$ of reads to be sampled in order to obtain a rigorous approximation of the frequent $k$-mers. 
	Let $k$ be a positive integer and $\D$ be a bag of $n$ reads.  Define the domain $X$ as the set of integers $\{1, \dots, n\}$, where every $i \in X$ corresponds to the $i$-th read of $\D$. Then define the family of real-valued functions $\mathcal{F} = \{ f_K , \forall K \in \Sigma^k\}$ where, for every $i \in X$ and for every $f_K \in \mathcal{F}$, the function $f_K(i)$ is the number of distinct positions in read $r_i$ where $k$-mer $K$ appears divided by the average size of the multiset of $k$-mers that appear in a read of $\D$: $f_K(i) = \sum_{j=0}^{n_i-k} \frac{  \phi_{r_i,K}(j)}{\ell_{\D,k}}$. Therefore $f_K(i) \in [0 ,\frac{\ell_{\max,\D,k}}{\ell_{\D,k}} ]$.  For each $f_K \in \mathcal{F}$, the subset of $X' = X \times [0 ,\frac{\ell_{\max,\D,k}}{\ell_{\D,k}} ]$ defined as $R_{f_K} = \{ (i,t) : t \leq f_K(i) \}$ is the associated range. Let $\mathcal{F}^+ = \{R_{f_K} , f_K \in \mathcal{F}\}$ be the range set on $X'$, and its corresponding range space $Q'$ be $Q'=(X',\mathcal{F}^+)$. 

A trivial upper bound to $PD(X , \mathcal{F})$ is given by $PD(X , \mathcal{F}) \leq \lfloor \log_2 |\mathcal{F}| \rfloor =\lfloor  \log_2 \sigma^k \rfloor$. The following result provides an improved upper bound to $PD(X , \mathcal{F})$ (the proof is in Supplemental Material Section~\ref{sec:advancedbounds} - see Proposition~\ref{prop:pseudobound_appendix}).

\begin{proposition}
	\label{prop:pseudobound}
	Let $\D$ be a bag of $n$ reads, $k$ a positive integer, $X= \{1, \dots, n\}$ be the domain, and let the family $\mathcal{F}$ of real-valued functions be $\mathcal{F} = \{ f_K , \forall K \in \Sigma^k\}$. Then the pseudodimension $PD(X , \mathcal{F})$ satisfies 
$ PD(X , \mathcal{F}) \leq \lfloor \log_2(\ell_{max,\D,k}) \rfloor + 1$.
\end{proposition}

Combining Proposition~\ref{prop:pseudosamplesize} and Proposition~\ref{prop:pseudobound}, we derive the following (see Supplemental Material Section~\ref{sec:advancedbounds} for the full analysis).

\begin{proposition}
	\label{prop:samplebound2}
	Let $S$ be a sample of $m$ reads from $\D$. For fixed threshold $\theta \in (0,1]$, error parameter $\varepsilon \in (0,\theta)$, and confidence parameter $\delta \in (0,1)$, if 
$		m \geq \frac{2}{\varepsilon^2} \left( \frac{\ell_{\max,\D,k}}{\ell_{\D,k}}\right)^2 \left( \lfloor \log_2\min( 2 \ell_{\max,\D,k} , \sigma^k ) \rfloor+\ln\left(\frac{1}{\delta} \right) \right)$
	then, with probability $\ge 1- \delta$, $FK(S,k,\theta-\varepsilon/2)$ is an $\varepsilon$-approximation of $FK(\D,k,\theta)$.
\end{proposition}

This bound significantly improves on the one in Proposition~\ref{prop:samplebound1}, since the factor $\ln(2 \sigma^k)$ is reduced to $\lfloor \log_2\min( 2 \ell_{\max,\D,k} , \sigma^k ) \rfloor$. However, even the bound from Proposition~\ref{prop:samplebound2} results in a sample size $m$ larger than $|\D|$. In the following section we proposes a method to further reduce the sample size $m$, which results in a practical sampling approach.

\section{\algname: Sampling Reads Algorithm to Estimate Frequent $k$-mers}
\label{sec:pseudo_alg_bags}

We now introduce \algname, which approximates the frequent $k$-mers by sampling \emph{bags of reads}.
We define $I_{\ell} = \{i_1 , i_2 , \dots , i_{\ell} \}$ as a \emph{bag} of $\ell$ indexes of reads of $\D$ chosen uniformly at random, with replacement, from the set $\{1,\dots,n\}$. Then we define an $\ell$-\emph{reads sample} $S_\ell$ as a collection of $m$ bags of $\ell$ reads $S_\ell = \{I_{\ell,1} , \dots , I_{\ell,m} \}$.


	Let $k$ be a positive integer and $\D$ be a bag of $n$ reads.  Define the domain $X$ as the set of bags of $\ell$ indexes of reads of $\D$.  Then define the family of real-valued functions $\mathcal{F} = \{ f_{K,\ell}, \forall K \in \Sigma^k\}$ where, for every $I_\ell \in X$ and for every $f_{K,\ell} \in \mathcal{F}$, we have $f_{K,\ell}(I_\ell) = \min(1, o_{I_{\ell}}(K)) / (\ell \ell_{\D,k})$, where $o_{I_{\ell}}(K) = \sum_{i \in I_{\ell}} \sum_{j=0}^{n_i-k} \phi_{r_i,K}(j) $ counts the number of occurrences of $K$ in all the $\ell$ reads of $I_{\ell}$.
	Therefore $f_{K,\ell}(I_\ell) \in \{0 ,\frac{1}{\ell \ell_{\D,k}}\}$ $\forall f_{K,\ell}$ and $\forall I_\ell$.  
	For each $f_{K,\ell} \in \mathcal{F}$, the subset of $X' = X \times \{0 ,\frac{1}{\ell \ell_{\D,k}}\}$ defined as $R_{f_{K,\ell}} = \{ (I_\ell,t) : t \leq f_{K,\ell}(I_\ell) \}$ is the associated range. Let $\mathcal{F}^+ = \{R_{f_{K,\ell}} , f_{K,\ell} \in \mathcal{F}\}$ be the range set on $X'$, and its corresponding range space $Q'$ be $Q'=(X',\mathcal{F}^+)$. 

Note that, for a given bag $I_{\ell}$, the functions $f_{K,\ell}$ are then biased if $K$ appears more than $1$ times in all the $\ell$ reads of $I_{\ell}$. 
We prove the following upper bound to the pseudodimension $PD(X , \mathcal{F})$ (see Proposition \ref{prop:pseudobound_bags} of Supplemental Material Section~\ref{appx:sec:advancedbounds_bags}).

\begin{proposition}
	\label{prop:pseudobound_bags_maintext}
	The pseudodimension $PD(X , \mathcal{F})$ satisfies $PD(X , \mathcal{F}) \leq \lfloor \text{log}_2(\ell \ell_{max,\D,k}) \rfloor + 1$.
\end{proposition}

We define the frequency $\hat{f}_{S_\ell}(K)$ of a $k$-mer $K$ obtained from the sample $S_\ell$ of bags of reads as  
$
\hat{f}_{S_\ell}(K) = \frac{1}{m}\sum_{I_{\ell,i} \in S_\ell} f_{K,\ell}(I_{\ell,i}).
$
Note that $\hat{f}_{S_\ell}(K)$ is a ``biased" version of 
$
f_{S_\ell}(K) = \frac{1}{m}\sum_{I_{\ell,i} \in S_\ell} o_{I_{\ell}}(K) / (\ell \ell_{\D,k}),
$
 which is an unbiased estimator of $f_\D(K)$ (i.e., $\E[f_{S_\ell}(K)] = f_\D(K)$). 


The following is our main technical results, and establishes a rigorous relation between the number  $m$ of bags of $\ell$ reads and the guarantees obtained by approximating the frequency $f_\D(K)$ of a $k$-mer $K$ with its (biased) estimate $\hat{f}_{S_\ell}(K)$. (The full analysis is in Supplemental Material Section~\ref{appx:sec:advancedbounds_bags} - see  Proposition~\ref{prop:pseudo_reads_samplesize_bags_appendix}.)

\begin{proposition}
	\label{prop:sample_bound3}
	Let $k$ and $\ell$ be two positive integers. Consider a sample $S_\ell$ of $m$ bags of $\ell$ reads from $\D$. For fixed frequency threshold $\theta \in (0,1]$, error parameter $\varepsilon \in (0,\theta)$, and confidence parameter $\delta \in (0,1)$, if 
	\begin{equation}
		\label{eq:sample_size3}
		m \geq \frac{2}{\varepsilon^2} \left( \frac{1}{\ell\ell_{\D,k}}\right)^2 \left( \lfloor \log_2\min( 2 \ell\ell_{\max,\D,k} , \sigma^k ) \rfloor+\ln\left(\frac{1}{\delta} \right) \right) 
	\end{equation}
then, with probability at least $1 - \delta$: 
\begin{packed_itemize}
	\item for any $k$-mer $K \in FK(\D,k,\theta)$ such that  $f_\D(A) \geq \tilde{\theta} = \frac{\ell_{\max,\D,k}}{\ell_{\D,k}}(1- (1- \ell \ell_{\D,k} \theta)^{1/\ell})$ it holds $\hat{f}_{S_\ell}(K)  \geq \theta - \varepsilon / 2$;
	\item for any $k$-mer $K$ with $\hat{f}_{S_\ell}(K) \geq \theta - \varepsilon/2$ it holds $f_\D(K) \geq \theta - \varepsilon$;
	\item for any $k$-mer $K \in FK(\D,k,\theta)$ it holds $f_\D(K) \geq \hat{f}_{S_\ell}(K) - \varepsilon / 2$;
	\item for any $k$-mer $K$ with $ \ell \ell_{\D,k}(\hat{f}_{S_\ell}(K) +  \varepsilon / 2) \leq 1 $ it holds $f_\D(K) \leq \frac{\ell_{\max,\D,k}}{\ell_{\D,k}}(1 - (1-\ell \ell_{\D,k} (\hat{f}_{S_\ell}(K) +  \varepsilon / 2))^{(1/\ell)})$.
\end{packed_itemize}
\end{proposition}

Given a sample $S_\ell$ of $m$ bags of $\ell$ reads from $\D$, with $m$ satisfying the condition of Proposition~\ref{prop:sample_bound3}, the set $A = \{(K,f_{S_\ell}(K)) : \hat{f}_{S_\ell}(K) \geq \theta - \varepsilon/2\}$ is \emph{almost} an $\varepsilon$-approximation of $FK(\D,k,\theta)$: Proposition \ref{prop:sample_bound3} ensures that all $k$-mers in $A$ have frequency $f_{\D}(K) \geq \theta - \varepsilon$ with probability at least $1-\delta$, but it does not guarantee that all $k$-mers with frequency $\in [\theta , \tilde{\theta})$ will be in output. However, we show in Section~\ref{sec:exp1} that, in practice, almost all of them are reported by \algname. We further remark that the derivations of \cite{sakeima} to obtain tight confidence intervals for $f_\D(A)$ using multiple values of $\ell$ are relevant also for the sampling scheme we employ in \algname; we will extend our analysis in this direction in the full version of this work.

\begin{wrapfigure}[17]{l}{0.65\textwidth}
	\begin{center}
		\vspace{-0.8cm}
		\begin{algorithm}[H]
			\label{alg:approximation_exact}
			\KwData{$\D$, $k$, $\theta \in (0,1]$, $\delta \in (0,1)$, $\varepsilon \in (0,\theta)$, integer $\ell \geq 1$}
			\KwResult{Approximation $A$ of $FK(\D,k,\theta)$ with probability at least $1-\delta$}
			$m \leftarrow \lceil \frac{2}{\varepsilon^2} \left( \frac{1}{\ell\ell_{\D,k}}\right)^2 \left( \lfloor \log_2\min( 2 \ell\ell_{\max,\D,k} , \sigma^k ) \rfloor+\ln\left(\frac{1}{\delta} \right) \right) \rceil$\label{line:computem}\;
			$S \leftarrow$ sample of exactly $m \ell$ reads drawn from $\D$\label{line:sample}\;
			$T \leftarrow \texttt{exact\_counting}(S,k)$\label{line:exactcount}\;
			$S_{\ell} \leftarrow $ random partition of $S$ into $m$ bags of $\ell$ reads each\label{line:createbags}\;
			$A \leftarrow \emptyset$\;
			\ForAll{$(K ,o_{S}(K)) \in T$}  
			{
				$S_K \leftarrow $ number of bags of $S_{\ell}$ where $K$ appears\label{line:compute_sk}\;
				$\hat{f}_{S_{\ell}}(K) \leftarrow S_K/(m \ell \ell_{\D,k})$\label{line:replace2}\;
				$f_{S_{\ell}}(K) \leftarrow o_{S}(K)/(m \ell \ell_{\D,k})$ 
				\label{alg:unbiasedfreq}\;
				\lIf{$\hat{f}_{S_{\ell}}(K) \geq \theta - \varepsilon/2$}{$A \leftarrow A \cup (K,f_{S_{\ell}}(K))$}
			}
			\Return $A$\;
			\caption{\algname$(\D, k, \theta, \delta, \varepsilon,\ell)$}
		\end{algorithm}
	\end{center}
\end{wrapfigure}

Our algorithm~\algname\ (Alg.~\ref{alg:approximation_exact}) builds on Proposition~\ref{prop:sample_bound3}, and returns the approximation of $FK(\D,k,\theta)$ defined by the set $A = \{(K,f_{S_\ell}(K)) : \hat{f}_{S_\ell}(K) \geq \theta - \varepsilon/2\}$. Therefore, with probability at least $1-\delta$ the output of \algname\ provides the guarantees stated in Proposition~\ref{prop:sample_bound3}.

\algname\ starts by computing the number $m$  of bags of $\ell$ reads as in Eq.~\ref{eq:sample_size3}, based on the input parameters $k,\theta,\delta,\varepsilon,\ell$ and on the characteristics ($\ell_{\D,k},\ell_{\max,\D,k}, \sigma$) of data\-set $\D$. It then draws a sample $S$ of exactly $m \ell$ reads, uniformly and independently at random, from $\D$ (with replacement). Next, it computes for each $k$-mer $K$ the number of occurrences $o_{S}(K)$ of $K$ in sample $S$, using any exact $k$-mers counting  algorithm. We denote the call of this method by \texttt{exact\_counting}$(S,k)$ (line~\ref{line:exactcount}), which returns a collection $T$ of pairs $(K,o_{S}(K))$. 
The sample $S$ is then partitioned into $m$ bags, where each bag contains exactly $\ell$ reads (line~\ref{line:createbags}). For each $k$-mer $K$, \algname\ computes the biased frequency $\hat{f}_{S_{\ell}}(K)$ (line~\ref{line:replace2}) and the unbiased frequency $f_{S_{\ell}}(K)$ (line~\ref{alg:unbiasedfreq}), reporting in output only $k$-mers with biased frequency at least $\theta-\varepsilon/2$ (line~\ref{alg:unbiasedfreq}). Note that the estimated frequency of a $k$-mer $K$ reported in output is always given by the unbiased frequency $f_{S_{\ell}}(K)$.
In practice the partition of $S$ into $m$ bags (line~\ref{line:createbags}) and the computation of $S_K$ (line~\ref{line:compute_sk}) could be high demanding in terms of running time and space, since one has to compute and store, for each $k$-mer $K$, the exact number $S_K$ of bags where $K$ appears at least once among all reads of the bag. 

We now describe an alternative, and much more efficient, approach to approximate the values $S_K$, without the need to explicitly compute the bags (line~\ref{line:createbags}). The number of reads in a given bag where $K$ appears is well approximated by a Poisson distribution $Poisson(R[K]/m)$, where $R[K]$ is the number of reads of $S$ where $k$-mer $K$ appears at least once. Therefore, the number $S_K$ of bags where $K$ appears at least once is approximated by a binomial distribution $Binomial(m,1-e^{-R[K]/m})$. Thus, one can avoid to explicitely create the bags and to exactly count $S_K$ by removing line~\ref{line:createbags}, and replacing lines~\ref{line:compute_sk}  and \ref{line:replace2} with $``\hat{f}_{S_{\ell}}(K) \leftarrow Binomial(m,1-e^{-R[K]/m})/(m \ell \ell_{\D,k})"$. Corollary 5.11~of~\cite{mitzenmacher2017probability} guarantees that, by using this Poisson distribution to approximate $S_K$, the output of \algname\ satisfies the properties of Proposition~\ref{prop:sample_bound3} with probability at least $1-2\delta$. This leads to the replacement of $``\ln(1/\delta)"$ with $``\ln(2/\delta)"$ in line~\ref{line:computem}. However, this approach requires to compute, for each $k$-mer $K$, the number of reads $R[K]$ of $S$ where $k$-appears at least once. 
We believe such computation can be obtained with minimal effort within the implementation of most $k$-mer counters, but we now describe a simple idea to approximate $R[K]$. Since most $k$-mers appear at most once in a read, the number of reads $R[K]$ where a $k$-mer $K$ appears is well approximated by the number of occurrences $T[K]$ of $K$ in the sample $S$. Thus, we can replace lines~\ref{line:compute_sk}  and \ref{line:replace2} with $``\hat{f}_{S_{\ell}}(K) \leftarrow Binomial(m,1-e^{-T[K]/m})/(m \ell \ell_{\D,k})"$, which only requires the counts $T[K]$ obtained from the exact counting procedure \texttt{exact\_counting}$(S,k)$ of line \ref{line:exactcount} (see Algorithm~\ref{alg:approximation_appendix} in Supplement Material). Note that approximating $R[K]$ with $T[K]$ leads to overestimate frequencies of few $k$-mers who reside in very repetitive sequences, e.g. $k$-mers composed by the same $k$ consecutive nucleotides, for which $T[K] \gg R[K]$. However, since the majority of $k$-mers reside in non-repetitive sequences, we can assume  $R[K] \approx T[K]$.

\section{Experimental Evaluation}
In this section we present results of our experimental evaluation. In particular:
\begin{packed_itemize}
	\item We assess the performance of \algname\ in approximating the set of frequent $k$-mers from a dataset of reads. In particular, we evaluate the accuracy of estimated frequencies and false negatives in the approximation, and compare \algname\ with the state-of-the-art sampling algorithm \sakeima~\cite{sakeima} in terms of sample size and running time.
	\item We evaluate \algname's performance for the comparison of metagenomic datasets. We use \algname's approximations to estimate abundance based distances (e.g., the Bray-Curtis distance) between metagenomic datasets, and show that the estimated distances can be used to obtain informative clusterings of metagenomic datasets (from the Sorcerer II Global Ocean Sampling Expedition \cite{GOS}\footnote{\url{https://www.imicrobe.us}}) in a fraction of the time required by the exact distances computation (i.e., based on exact $k$-mers frequencies).
	\item We test \algname\ to discover discriminative $k$-mers between pairs of datasets. We show that \algname\ identifies almost all discriminative $k$-mers from pairs of metagenomic datasets from~\cite{liu2017unbiased} and the Human Microbiome Project (HMP)\footnote{\url{https://hmpdacc.org/HMASM/}}, with a significant speed-up compared to standard approaches.
	\end{packed_itemize}

\subsection{Implementation, Datasets, Parameters, and Environment}
We implemented \algname\ as a combination of a Python script, which performs the reads sampling and saves the sample on a file, and C++, as a modification of \kmc~3~\cite{kokot2017kmc}\footnote{Available at \url{https://github.com/refresh-bio/KMC}}, a fast and efficient counting $k$-mers algorithm. Note that our flexible sampling technique can be combined with any $k$-mer counting algorithm. (See Supplemental Material for results, e.g. Figure~\ref{fig:runningtimes_withJelly}, obtained using \jellyfish~v.~\texttt{2.3}\footnote{Available at \url{https://github.com/gmarcais/Jellyfish}} as $k$-mer counter in \algname). 
We use the variant of \algname\ that employs the Poisson approximation for computing $S_{K}$ (see end of Section~\ref{sec:pseudo_alg_bags}). 
\algname\ implementation and scripts for reproducing all results are publicity available\footnote{Available at \url{https://github.com/VandinLab/SPRISS}}. 
We compared \algname\ with the exact $k$-mer counter \kmc\ and with \sakeima~\cite{sakeima}\footnote{Available at \url{https://github.com/VandinLab/SAKEIMA}}, the state-of-the-art sampling-based algorithm for approximating frequent $k$-mers. In all experiments we fix $\delta=0.1$ and $\varepsilon = \theta - 2/t_{\D,k}$. If not stated otherwise, we considered $k=31$ and $\ell = \lfloor 0.9/(\theta \ell_{\D,k}) \rfloor$ in our experiments. 
When comparing running times, we did not consider the time required by \algname\ to materialize the sample in a file, since this step is not explicitly performed in \sakeima\ and could be easily done at the time of creation of the reads dataset. 
For \sakeima, as suggested in~\cite{sakeima} we set the number $\ell_{SK}$ of $k$-mers in a bag to be $\ell_{SK}= \lfloor 0.9/\theta \rfloor$. We remark that a bag of reads of \algname\ contains the same (expected) number of $k$-mers positions of a bag of \sakeima; this guarantees that both algorithms provide outputs with the same guarantees, thus making the comparison between the two methods fair.
To assess \algname\ in approximating frequent $k$-mers, we considered $6$ large metagenomic datasets
from HMP, each with $\approx 10^8$ reads and average read length $\approx 100$ (see Supplemental Table \ref{tab:dataHMP}). 
For the evaluation of \algname\ in comparing metagenomic datasets, we also used $37$ small metagenomic datasets from the Sorcerer II Global Ocean Sampling Expedition \cite{GOS}, each with $\approx 10^4\text{-}10^5$ reads and average read length $\approx 1000$ (see Supplement Table~\ref{tab:dataGOS}). For the assessment of \algname\ in the discovery of discriminative $k$-mers we used two large datasets from~\cite{liu2017unbiased}, \texttt{B73} and \texttt{Mo17}, each with $\approx 4 \cdot 10^8$ reads and average read length $ = 250$ (see Supplemental Table~\ref{tab:dataMOB}), and we also experimented with the HMP datasets.
All experiments have been performed on a machine with 512 GB of RAM and 2 Intel(R) Xeon(R) CPU E5-2698 v3 @2.3GHz, with one worker, if not stated otherwise. All reported results  are averages over $5$ runs. 

\subsection{Approximation of Frequent $k$-mers}
\label{sec:exp1}
In this section we first assess the quality of the approximation of $FK(\D,k,\theta)$ provided by \algname, and then compare \algname\ with \sakeima.

We use \algname\ to extract approximations of frequent $k$-mers on 6 datasets from HMP for values of the minimum frequency threshold $\theta\in \{2.5\cdot10^{-8}, 5\cdot10^{-8}, 7.5\cdot10^{-8},10^{-7}\}$.
The output of \algname\ satisfied the guarantees from Proposition \ref{prop:sample_bound3} for all 5 runs of every combination of dataset and $\theta$.
In all cases the estimated frequencies provided by \algname\ are close to the exact ones (see Figure~\ref{fig:frequencies_bounds_SPRISS} for an example). In fact, the average (across all reported $k$-mers) absolute deviation of  the estimated frequency w.r.t. the true frequency is always small, i.e. one order of magnitude smaller than $\theta$ (Figure~ \ref{fig:avg_dev}), and the maximum deviations is very small as well  (Figure~\ref{fig:max_dev}).
In addition, \algname\ results in a very low false negative rate (i.e., fraction of $k$-mers of $FK(\D,k,\theta)$ not reported by \algname), which is always been below $0.012$ in our experiments.

 In terms of running time, \algname\ required at most $64\%$ of the time required by the exact approach \kmc\ (Figure~\ref{fig:running_time}). This is due to \algname\ requiring to analyze at most $34\%$ of the entire dataset (Figure~\ref{fig:sample_size}).
Note that the use of collections of bags of reads is crucial to achieve useful sample size, i.e. lower than the whole dataset:
the sample size from Hoeffding's inequality and union bound (Proposition~\ref{prop:samplebound1}), and the one from pseudodimension without collections of bags (Proposition~\ref{prop:samplebound2}) are $\approx 10^{16}$ and $\approx 10^{15}$, respectively, which are useless for datasets of $\approx 10^{8}$ reads. 
These results show that \algname\ obtains very accurate approximations of frequent $k$-mers in a fraction of the time required by exact counting approaches.

\begin{figure}[h!]
	\centering
	\subfloat[]{\includegraphics[width=.45\linewidth]{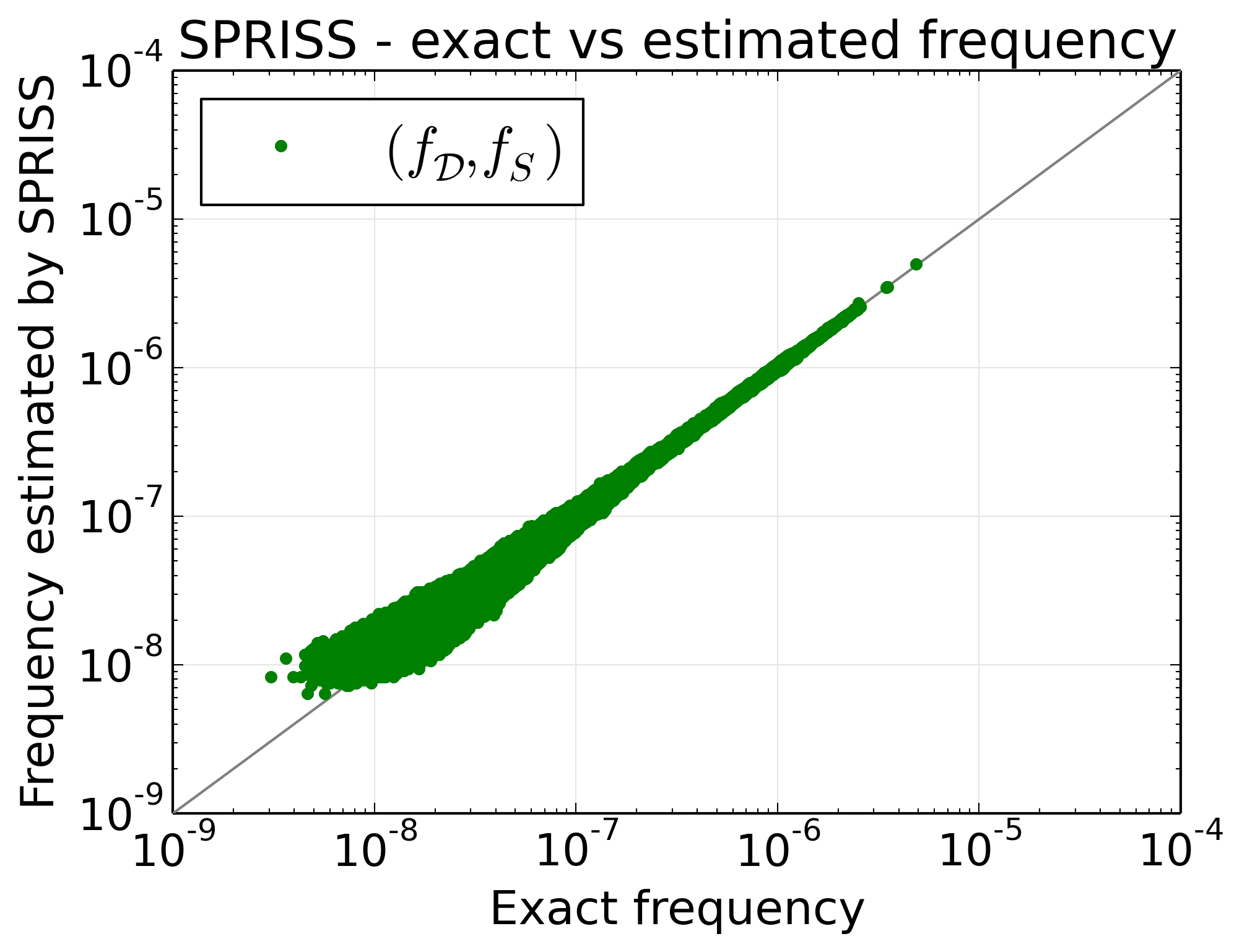} \label{fig:frequencies_bounds_SPRISS}}
	\subfloat[]{\includegraphics[width=.435\linewidth]{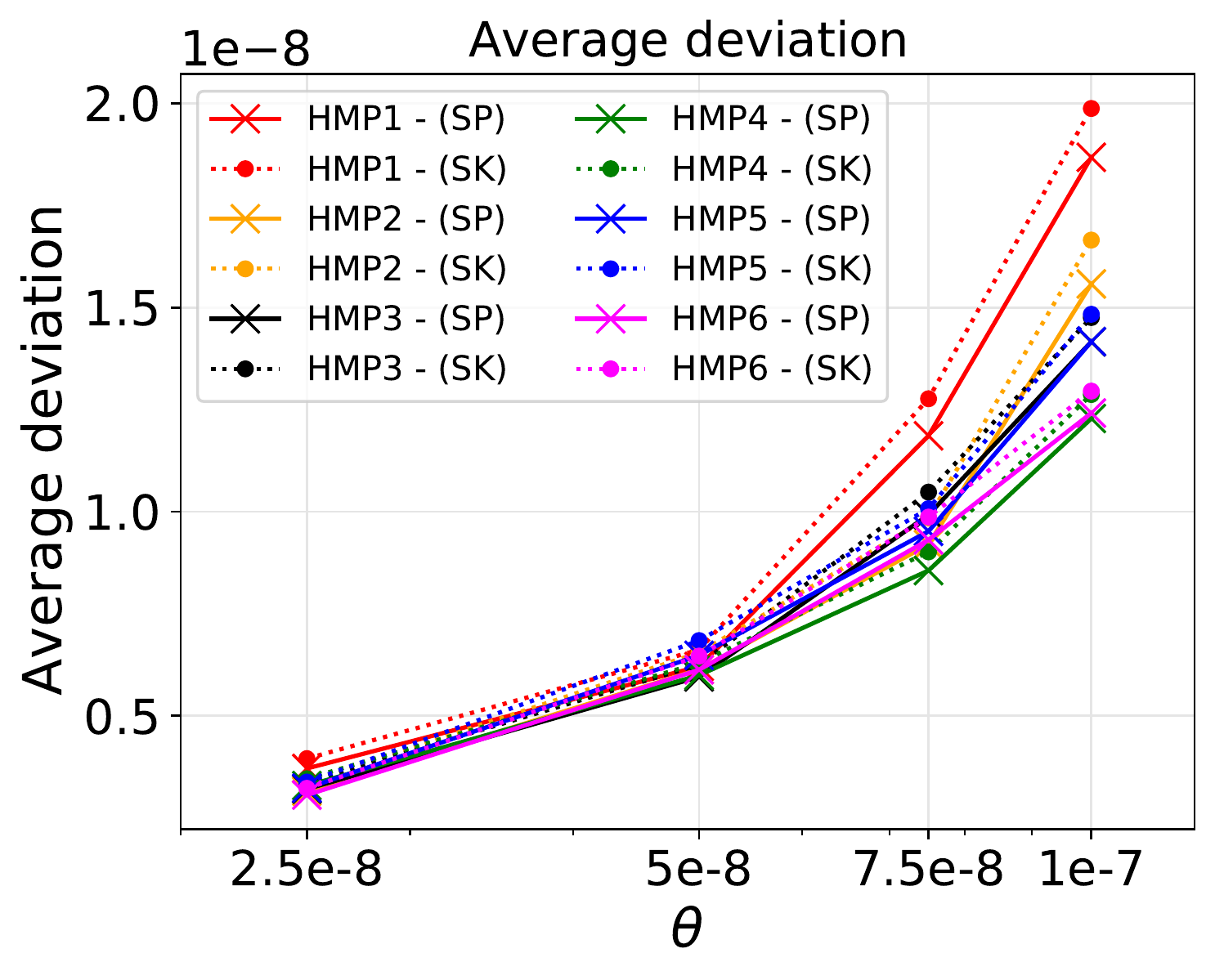} \label{fig:avg_dev}}\\
	\subfloat[]{\includegraphics[width=.45\linewidth]{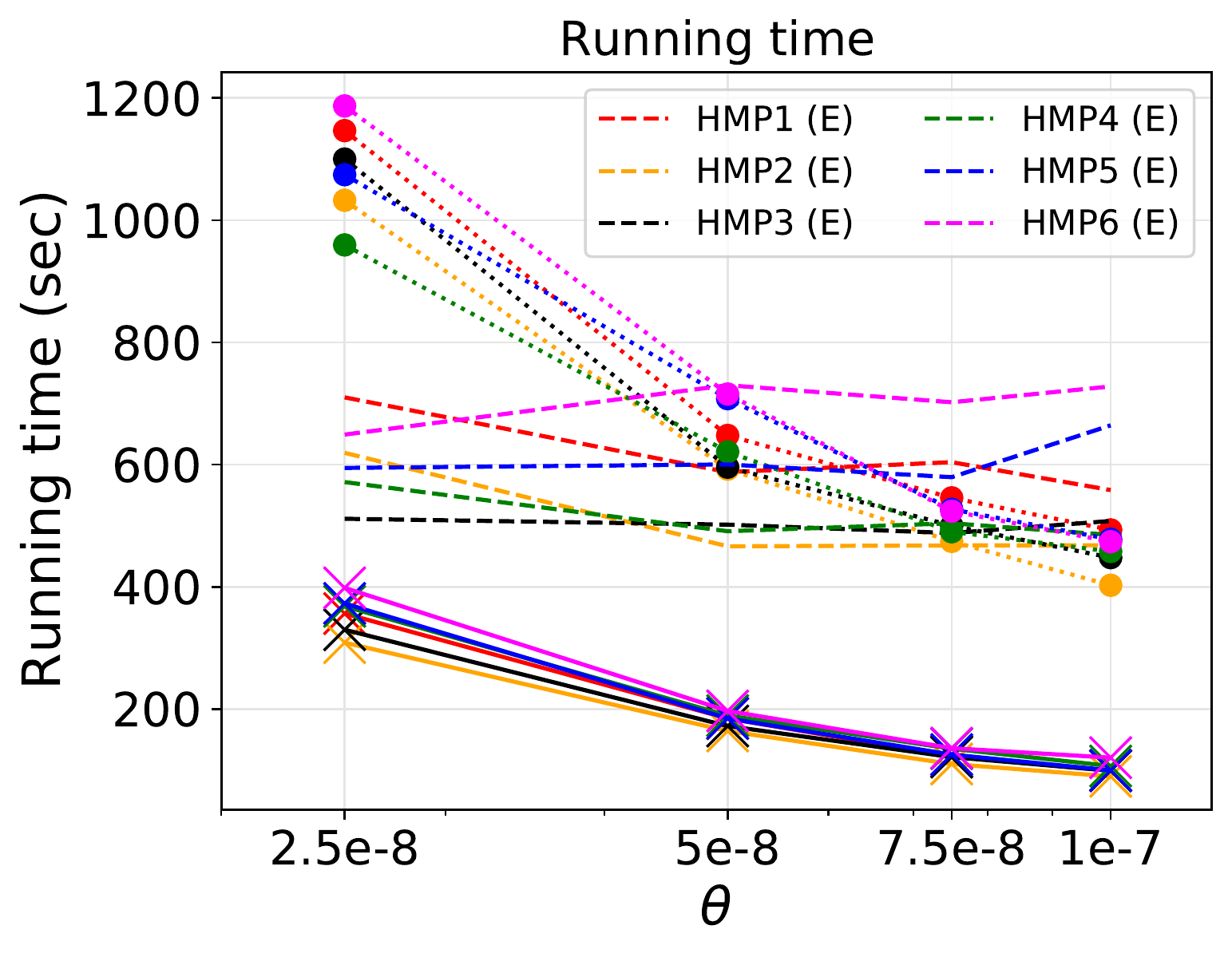} \label{fig:running_time}}
	\subfloat[]{\includegraphics[width=.435\linewidth]{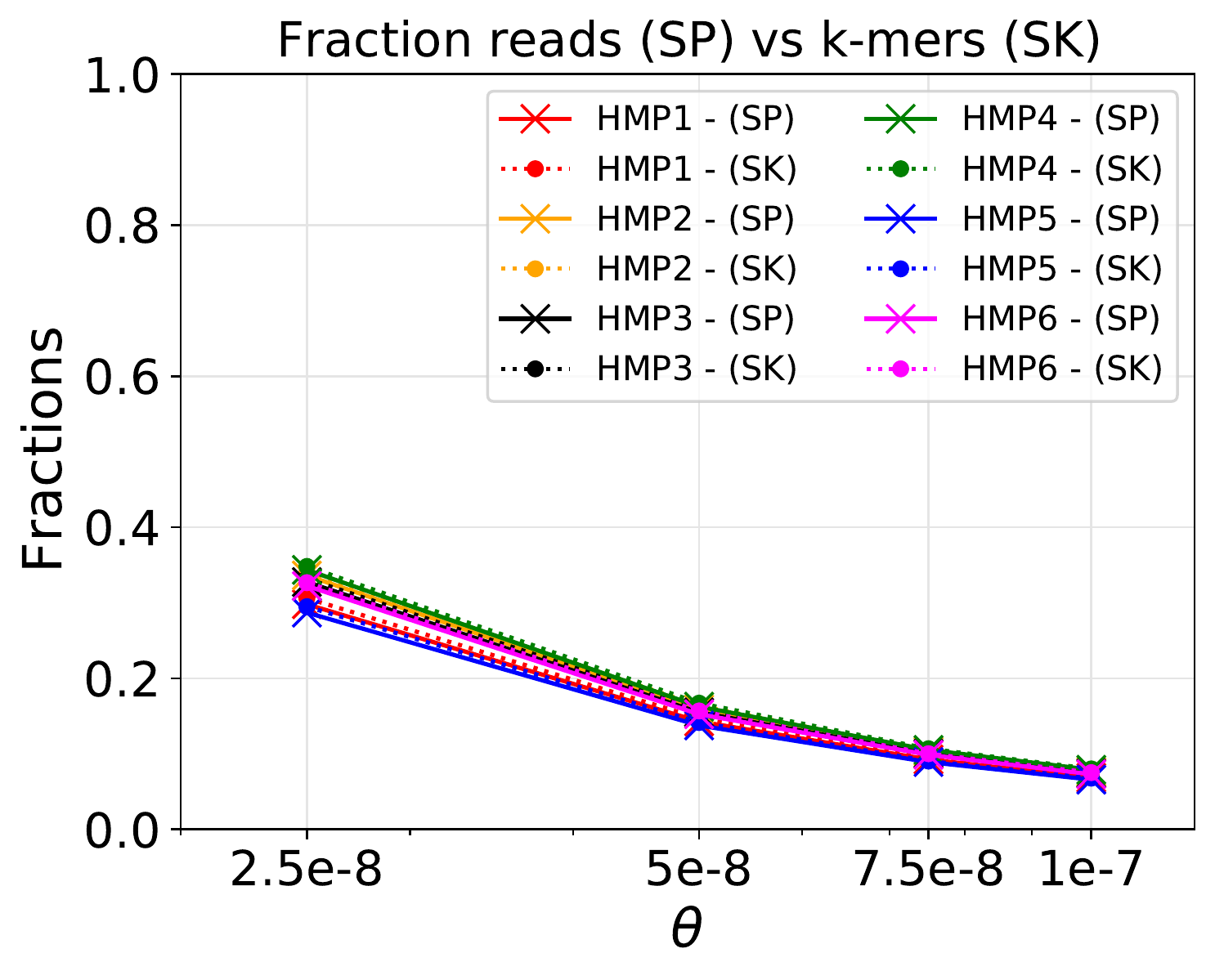} \label{fig:sample_size}}\\
	\caption{ (a) $k$-mers exact frequency and frequency estimated by \algname\ for dataset \texttt{SRS024075} and $\theta = 2.5\cdot10^{-8}$. (b) Average  deviations between exact frequencies and frequencies estimated by \algname\ (\algnameshort) and \sakeima\ (\sakeimashort), for various datasets and values of $\theta$. (c) Running time of \algname\ (\algnameshort), \sakeima\ (\sakeimashort), and the exact computation (\exactshort) - see also legend of panel~\ref{fig:avg_dev}). (d) Fraction of the dataset analyzed by \algname\ (\algnameshort) and by \sakeima\ (\sakeimashort).}
\end{figure}

We then compared \algname\ with \sakeima. In terms of quality of approximation, \algname\ reports approximations with an average deviation lower than \sakeima's approximations, while \sakeima's approximations have a lower maximum deviation. However, the ratio between the maximum deviation of \algname\ and the one of \sakeima\ are always below 2. Overall, the quality of the approximation provided by \algname\ and \sakeima\ are, thus, comparable. In terms of running time, \algname\ significantly improves over \sakeima\ (Figure~\ref{fig:running_time}), and processes slightly smaller portions of the dataset compared to \sakeima\ (Figure~\ref{fig:sample_size}).
Summarizing, \algname\ is able to report most of the frequent $k$-mers and estimate their frequencies with small errors, by analyzing small samples of the datasets and with significant improvements on running times compared to exact approaches and to state-of-the-art sampling algorithms. 

\subsection{Comparing Metagenomic Datasets }
\label{sec:exp2}

We evaluated \algname\ to compare metagenomic datasets by computing an approximation to the Bray-Curtis (BC) distance between pairs of datasets of reads, and using such approximations to cluster datasets.

Let $\D_1$ and $\D_2$ be two datasets of reads. Let $\mathcal{F}_1 = FK(\D_1,k,\theta)$ and $\mathcal{F}_2 = FK(\D_2,k,\theta)$ be the set of frequent $k$-mers respectively of $\D_1$ and $\D_2$, where $\theta$ is  a minimum frequency threshold. 
The \emph{BC distance} between $\D_1$ and $\D_2$ considering only frequent $k$-mers is defined as
$BC(\D_1,\D_2,\mathcal{F}_1,\mathcal{F}_2) = 1 - 2I/U$, where $I = \sum_{K \in \mathcal{F}_1 \cap \mathcal{F}_2} \min\{o_{\D_1}(K),o_{\D_2}(K)\}$ and $U = \sum_{K \in \mathcal{F}_1} o_{\D_1}(K) + \sum_{K \in \mathcal{F}_2} o_{\D_2}(K).$ Conversely, the \emph{BC similarity} is defined as $1 - BC(\D_1,\D_2,\mathcal{F}_1,\mathcal{F}_2)$.

We considered 6 datasets from HMP, and estimated the BC distances among them by using \algname\ to approximate the sets of frequent $k$-mers $\mathcal{F}_1 = FK(\D_1,k,\theta)$ and $\mathcal{F}_2 = FK(\D_2,k,\theta)$ for the values of $\theta$ as in Section~\ref{sec:exp1}. We compared such estimated distances with the exact BC distances and with the estimates obtained using \sakeima. Both \algname\ and \sakeima\ provide accurate estimates of the BC distances (Figure~\ref{fig:legend_and_theta1} and Figure~\ref{fig:all_thetas}), which can be used to assess the relative similarity of pairs of datasets. However, to obtain such approximations  \algname\ requires at most $25\%$ of the time required by \sakeima\ and usually $30\%$ of the time required by the exact computation with \kmc (Figure~\ref{fig:all_running_times}). Therefore \algname\ provides accurate estimates of metagenomic distances in a fraction of time required by other approaches.

 \begin{figure}[h!]
	\centering
	\subfloat[]{\includegraphics[width=.44\linewidth]{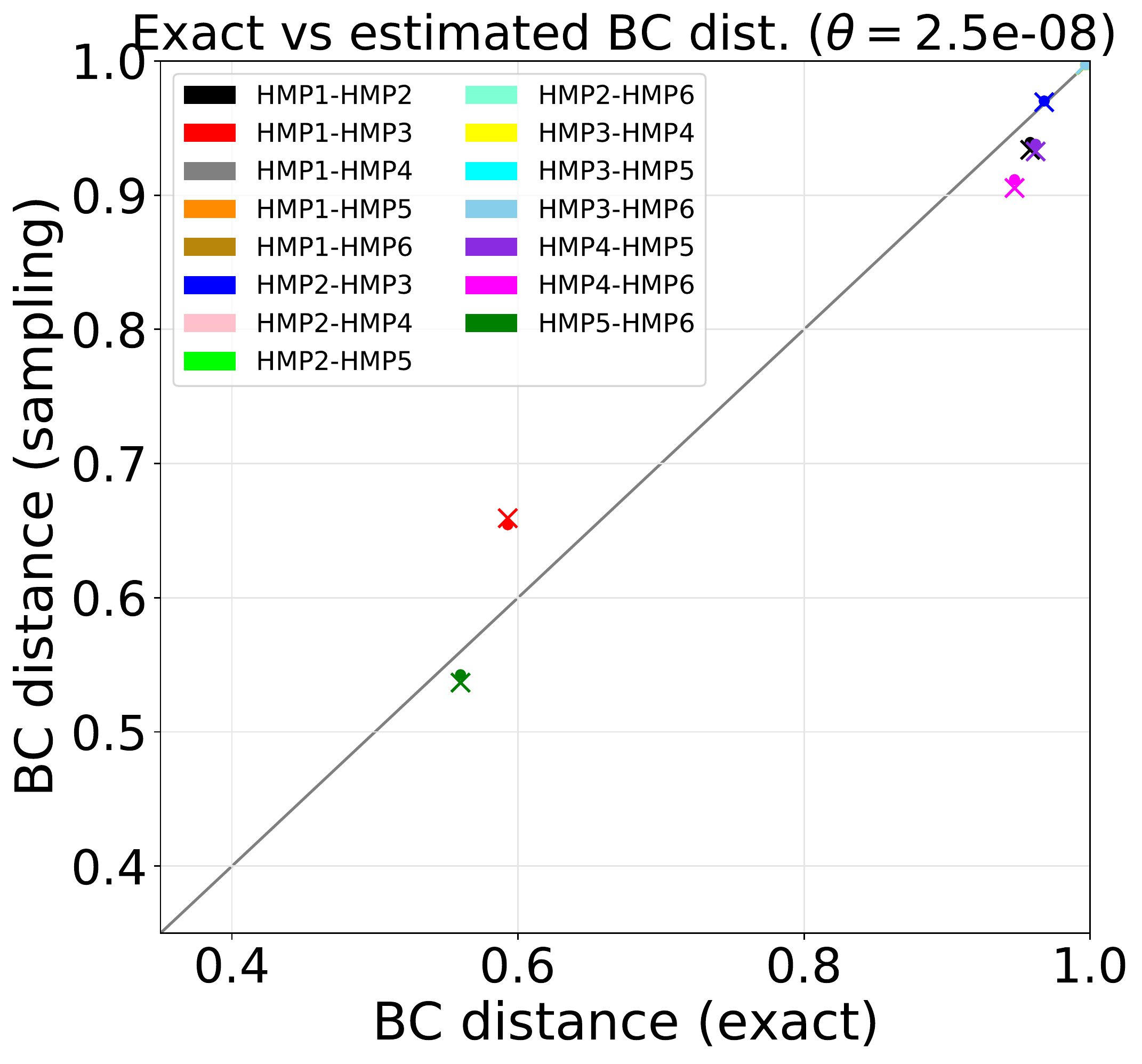} \label{fig:legend_and_theta1}}
	\subfloat[]{\includegraphics[width=.445\linewidth]{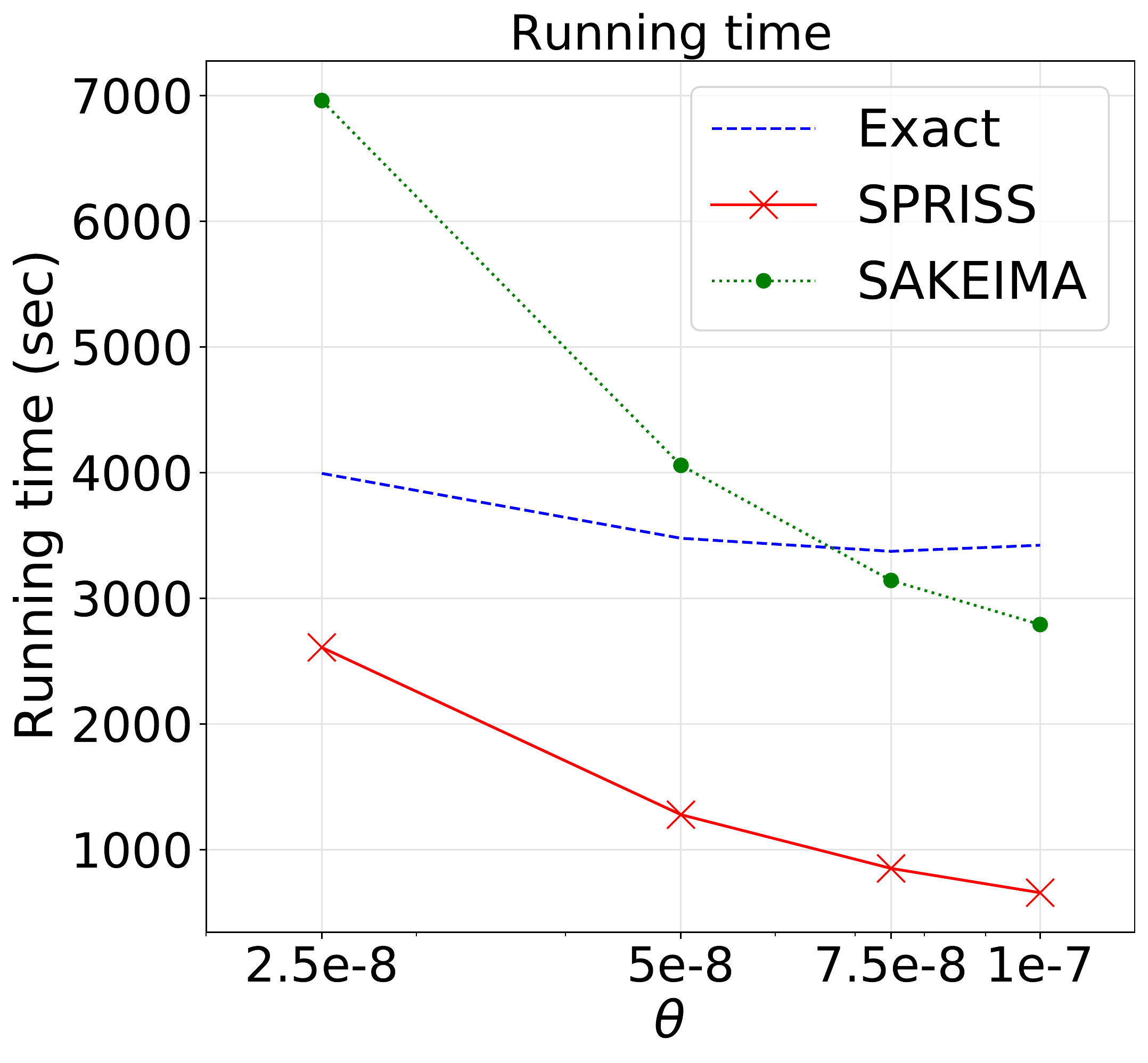}  \label{fig:all_running_times}}\\
	\subfloat[]{\includegraphics[width=.435\linewidth]{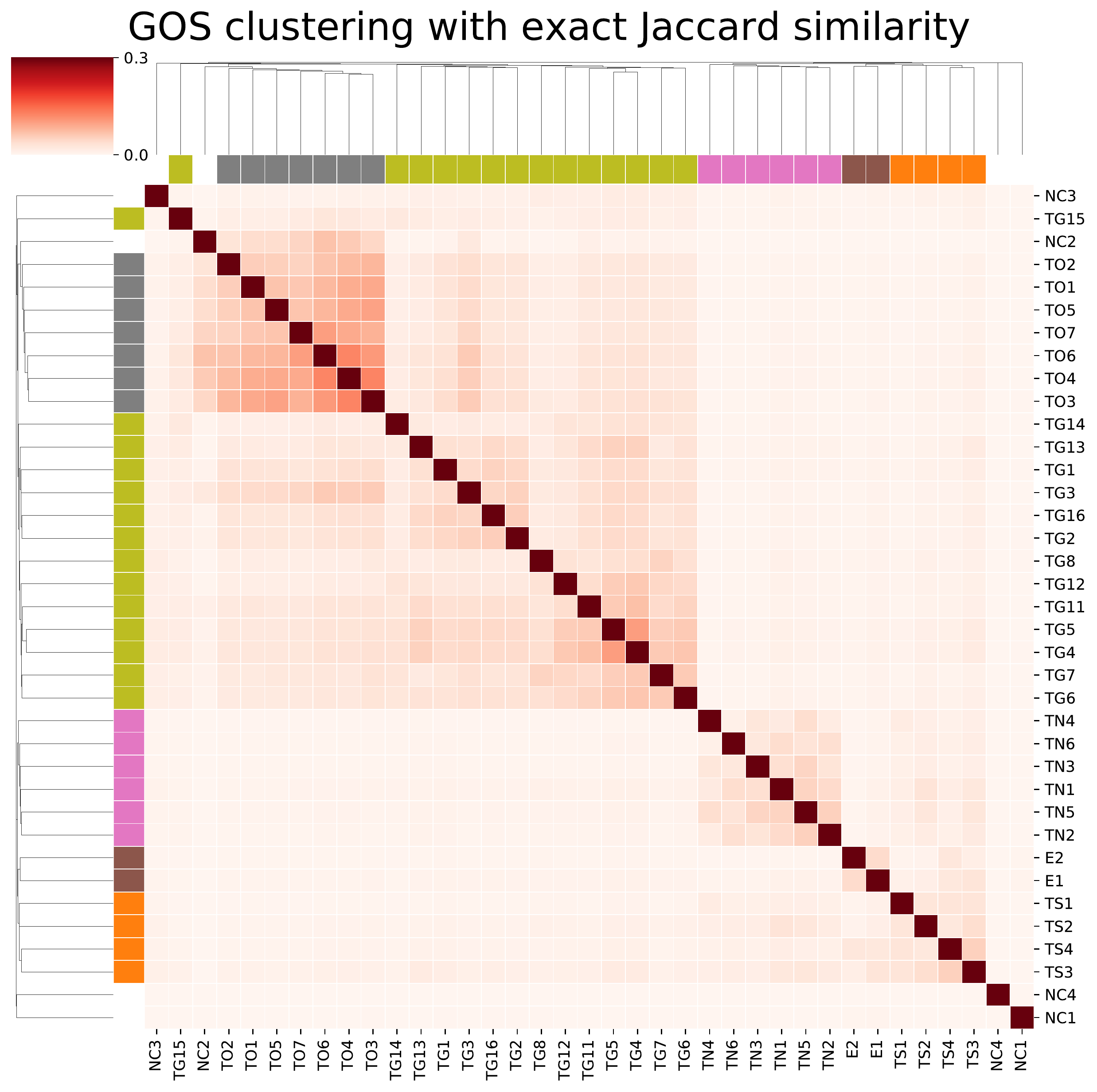}  \label{fig:jaccard_clustering}}
	\subfloat[]{\includegraphics[width=.435\linewidth]{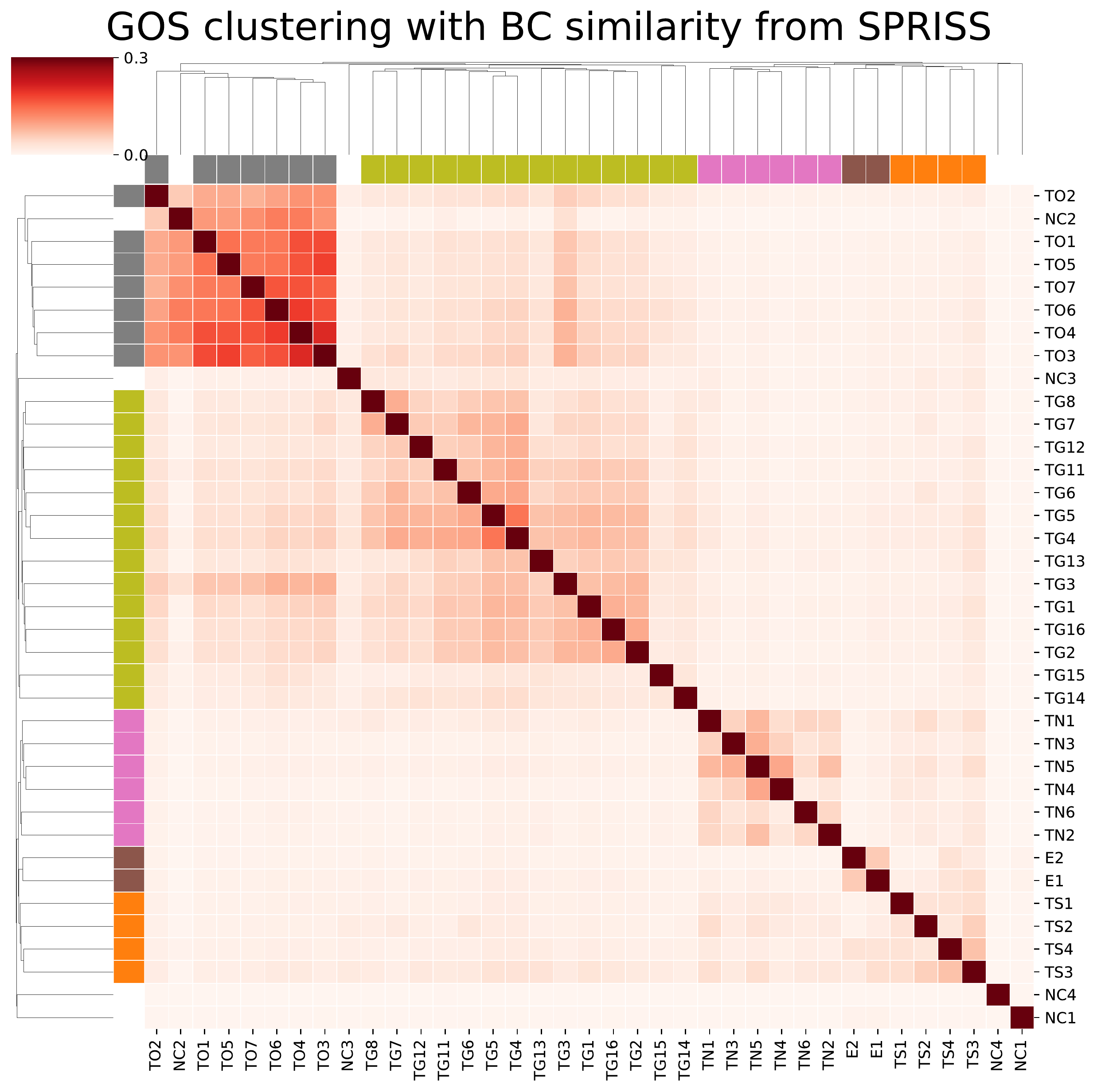} \label{fig:bc_sampling_clustering}}
	\caption{(a) Comparison of the approximations of the Bray-Curtis (BC) distances using approximations of frequent $k$-mers provided by \algname\ ($\times$) and by \sakeima\ ($\bullet$), and the exact distances, for $\theta=  2.5\cdot10^{-8}$. (b) Running time to approximate BC distances for all pairs of datasets with \algname, with \sakeima, and the exact approach. (c) Average linkage hierarchical clustering of GOS datasets using Jaccard similarity. 
(d) Same as (c), using estimated BC similarity from \algname\ with $50\%$ of the data. (See also larger Figures~\ref{fig:JC_clustering_supp}-\ref{fig:BCS_clustering_supp} in Supplemental Material for better readability of datasets' labels and computed clusters.)}
	\label{fig:BC_approximation_comparison}
\end{figure}

As an example of the impact in accurately estimating distances among metagenomic datasets, we used the sampling approach of \algname\ to approximate all pairwise BC distances among $37$ small datasets from the Sorcerer II Global Ocean Sampling Expedition (GOS) \cite{GOS}, and used such distances to cluster the datasets using  average linkage hierarchical clustering. The $k$-mer based clustering of metagenomic datasets is often performed by using \emph{presence-based} distances, such as the Jaccard distance~\cite{ondov2016mash}, which estimates similarities between two datasets by computing the fraction of $k$-mers in common between the two datasets. Abundance-based distances, such as the BC distance~\cite{benoit2016multiple,danovaro2017submarine,dickson2017carryover}, 
provide more detailed measures based also on the $k$-mers abundance, but are often not used due to the heavy computational requirements to extract all $k$-mers counts. However, the sampling approach of \algname\ can significantly speed-up the computation of all BC distances, and, thus, the entire clustering analysis. In fact, for this experiment, the use of \algname\ reduces the time required to analyze the datasets (i.e., obtain $k$-mers frequencies, compute all pairwise distances, and obtain the clustering) by $62\%$.

We then compared the clustering obtained using the Jaccard distance (Figure~\ref{fig:jaccard_clustering}) and the clustering obtained using the estimates of the BC distances (Figure~\ref{fig:bc_sampling_clustering}) obtained using only $50\%$ of reads in the GOS datasets, which are assigned to groups and macro-groups according to the origin of the sample~\cite{GOS}.
Even if the BC distance is computed using only a sample of the datasets, while the Jaccard distance is computed using the entirety of all datasets, the use of approximate BC distances leads to a better clustering in terms of correspondence of clusters to groups,
 and to the correct cluster separation for macro-groups. 
In addition, the similarities among datasets in the same group
and the dissimilarities among datasets in different groups are more accentuated using the approximated BC distance. In fact, the ratio between the average BC similarity among datasets in the same group and the analogous average Jaccard is in the interval 
$[1.25,1.75]$ for all groups. In addition, the ratio between i) the difference of the average BC similarity within the tropical macro-group and the average BC similarity between the tropical and temperate groups, and ii) the analogous difference using the Jaccard similarity is $\approx 1.53$. These results tell us the approximate BC-distances, computed using only half of the reads in each dataset, increase by $\approx 50\%$ the similarity signal inside all groups defined by the original study \cite{GOS}, and the dissimilarities between the two macro-groups (tropical and temperate).
 
To conclude, the estimates of the BC similarities obtained using the sampling scheme of \algname\ allows to better cluster metagenomic datasets than using the Jaccard similarity, while requiring less than $40\%$ of the time needed by the exact computation of BC similarities, even for fairly small metagenomic datasets.

\subsection{Approximation of Discriminative $k$-mers}
\label{sec:exp3}
In this section we assess \algname\ for approximating discriminative $k$-mers in metagenomic datasets.
In particular, we consider the following definition of discriminative $k$-mers~\cite{liu2017unbiased}.
Given two datasets $\D_1,\D_2$, and a minimum frequency threshold $\theta$,  we define the set $DK(\D_1,\D_2,k,\theta,\rho)$ of $\D_1$-discriminative $k$-mers as the collection of $k$-mers $K$ for which the following conditions both hold: 1. $K \in FK(\D_1,k,\theta)$; 2. $f_{\D_1}(K) \geq \rho f_{\D_2}(K)$, with $\rho=2$.  Note that the computation of $DK(\D_1,\D_2,k,\theta,\rho)$ requires to extract $FK(\D_1,k,\theta)$ and $FK(\D_2,k,\theta/\rho)$.  \sloppy{\algname\ can be used to approximate the set $DK(\D_1,\D_2,k,\theta,\rho)$, by computing approximations $\overline{FK}(\D_i,k,\theta)$ of the sets $FK(\D_i,k,\theta)$, $i=1,2$, of frequent $k$-mers in $\D_1,\D_2$, and then reporting a $k$-mer $K$ as $\D_1$-discriminative if the following conditions both hold: 1. $K \in \overline{FK}(\D_1,k,\theta)$; 2. $K \notin \overline{FK}(\D_2,k,\theta)$, or $f_{S^1_\ell}(K) \geq \rho f_{S^2_\ell}(K)$ when $K \in \overline{FK}(\D_2,k,\theta)$.}

To evaluate such approach, we considered two datasets from~\cite{liu2017unbiased}, and $\theta=2 \cdot 10^{-7}$  and $\rho=2$, which are the parameters used in~\cite{liu2017unbiased}. We used the sampling approach of \algname\ with $\ell = \lfloor 0.02/(\theta \ell_{\D,k}) \rfloor$ and $\ell = \lfloor 0.04/(\theta \ell_{\D,k}) \rfloor$, resulting in analyzing of $5\%$ and $10\%$ of all reads, to approximate the sets of discriminative  $\D_1$-discriminative and of $\D_2$-discriminative $k$-mers. When $5\%$ of the reads are used, the false negative rate is $<0.028$, while when $10\%$ of the reads are used, the false negative rate is $<0.018$. The running times are $\approx 1130$ sec. and $\approx 1970$ sec., respectively, while the exact computation of the discriminative $k$-mers with \kmc\ requires $\approx 10^4$ sec. (we used 32 workers for both \algname\ and \kmc).  Similar results are obtained when analyzing pairs of HMP datasets, for various values of $\theta$ (Figure~\ref{fig:results_discriminative_HMP}).
These results show that \algname\ can identify discriminative $k$-mers with small false negative rates while providing a remarkable improvement in running time compared to the exact approach. 

\section{Conclusions}

We presented \algname, an efficient algorithm to compute rigorous approximations of frequent $k$-mers and their frequencies by sampling reads. \algname\ builds on pseudodimension, an advanced concept from statistical learning theory. Our extensive experimental evaluation shows that \algname\ provides high-quality approximations and can be employed to speed-up exploratory analyses in various applications, such as the analysis of metagenomic datasets and the identification of discriminative $k$-mers. Overall, the sampling approach used by \algname\ provides an efficient way to obtain a representative subset of the data that can be used to perform complex analyses more efficiently than examining the whole data, while obtaining representative results.

\newpage
\appendix{ \huge Supplemental Material}

\setcounter{figure}{0}
\renewcommand{\thefigure}{S\arabic{figure}}

\section{Analysis of Simple Reads Sampling Algorithm}
\label{sec:A}

In this section we prove Proposition \ref{prop:samplebound1}, which here corresponds to  Proposition \ref{prop:unionboundsample}. To this aim, we need to introduce and prove some preliminary results.
\begin{proposition}
	The expectation $\E[t_{S,k}]$ of the size of the multiset of $k$-mers that appear in $S$ is $m \ell_{\D,k}$.
\end{proposition}
\begin{proof}
	Let $X(r_i) = n_i - k + 1$ be the number of starting positions for $k$-mers in read $r_i$ sampled uniformly at random form $\D$, $i \in \{1,\dots,n\}$. $\E[X(r_i)] = \sum_{r_i \in \D} \frac{1}{n} (n_i - k + 1) = \ell_{\D,k}$. Combining this with the linearity of the expectation, we have:
	\begin{equation*}
	\E[t_{S,k}] = \E \left[ \sum_{r_i \in S} (n_i - k + 1) \right] = \sum_{r_i \in S} \E[n_i - k + 1] = m \E[X(r_i)] = m \ell_{\D,k}.
	\end{equation*}
\end{proof}

Given a $k$-mer $K$, its support $o_S(K)$ in $S$ is defined as $o_S(K) = \sum_{r_i \in S} \sum_{j=0}^{n_i-k} \phi_{r_i,K}(j)$. We define the frequency of $K$ in $S$ as $f_S(K) = o_S(K) / (m \ell_{\D,k})$, that is the ratio between the support of $K$ and the expectation $\E[t_{S,k}] = m \ell_{\D,k}$ of the size of the multiset of $k$-mers that appear in $S$. This definition of $f_S(K)$ gives us an unbiased estimator for $f_\D(K)$.

\begin{proposition}
	\label{prop:unbiased}
	The frequency $f_S(K) = o_S(K) / (m \ell_{\D,k})$ is an unbiased estimator for $f_\D(K) = o_\D(K)/t_{\D,k}$.
\end{proposition}
\begin{proof}
	Let $X_{r_i}(K) = \sum_{j=0}^{n_i-k} \phi_{r_i,K}(j)$ be the number of distinct positions where $k$-mer $K$ appears in read $r_i$ sampled uniformly at random form $\D$, $i \in \{1,\dots,n\}$. 
	$E[ X_{r_i}(K)] = \sum_{r_i \in \D}\left( \frac{1}{n} \sum_{j=0}^{n_i-k} \phi_{r_i,K}(j) \right) = o_\D(K)/n.$ 
	Combining this with the linearity of the expectation, we have:
	\begin{equation*}
	\E[f_S(K)] = \frac{E[o_S(K)]}{m \ell_{\D,k}} =  \frac{\E[\sum_{r_i \in S} \sum_{j=0}^{n_i-k} \phi_{r_i,K}(j)]}{m \ell_{\D,k}}  = \frac{\E[ X_{r_i}(K)]}{\ell_{\D,k}} = \frac{o_\D(K)}{n \ell_{\D,k}} = \frac{o_\D(K)}{t_{\D,k}} = f_\D(K). 
	\end{equation*}
\end{proof}

By using the sampling framework based on reads and the Hoeffding inequality \cite{mitzenmacher2017probability}, we prove the following bound on the probability that $f_S(K)$ is not within $\varepsilon/2$ from $f_\D(K)$, for an arbitrary $k$-mer $K$. 

\begin{proposition}
	\label{prop:hoeffding}
	Consider a sample $S$ of $m$ reads from $\D$. Let $\ell_{\max,\D,k} = \max_{r_i \in \D}(n_i - k + 1)$. Let $K \in \Sigma^k$ be an arbitrary $k$-mer. For a fixed accuracy parameter $\varepsilon \in (0,1)$ we have:
	\begin{equation}
		\Pr \left( | f_S(K) - f_\D(K) | \geq \frac{\varepsilon}{2} \right) \leq 2 \exp\left( -   \frac{1}{2} m \varepsilon^2 \left( \frac{\ell_{\D,k} }{ \ell_{\max,\D,k}}\right)^2 \right).
	\end{equation}
\end{proposition}
\begin{proof}
	The frequency $f_S(K) = o_S(K) / (m \ell_{\D,k})$ of $K$ in $S$ can be rewritten as:
	\begin{equation}
		f_S(K) = \frac{ \sum_{r_i \in S} \sum_{j=0}^{n_i-k} \phi_{r_i,K}(j)}{m \ell_{\D,k}} = \sum_{r_i \in S} \sum_{j=0}^{n_i-k} \frac{\phi_{r_i,K}(j)}{m\ell_{\D,k}} = \sum_{r_i \in S}\hat{\phi}_{K}(r_i),
	\end{equation}
	where the random variable (r.v.) $\hat{\phi}_{K}(r_i) = \sum_{j=0}^{n_i-k} \frac{\phi_{r_i,K}(j)}{m \ell_{\D,k}}$ is the number of times $K$ appears in read $r_i$ divided by $m \ell_{\D,k}$. Thus, $f_S(K)$ can be rewritten as a sum of $m$ independent r.v. that take values in $[0,\frac{\ell_{\max,\D,k}}{m \ell_{\D,k}}]$. Combining this fact with Proposition~\ref{prop:unbiased}, and by applying the Hoeffding inequality \cite{mitzenmacher2017probability} we have:
	\begin{equation*}
		\Pr (|f_S(K) - f_\D(K) | \geq \frac{\varepsilon}{2}) \leq 2 \exp \left( \frac{- 2 (\varepsilon/2)^2}{m  \left(  \frac{\ell_{\max,\D,k}}{m \ell_{\D,k}} \right)^2 } \right) = 2 \exp \left(  - \frac{1}{2} m  \varepsilon^2 \left( \frac{\ell_{\D,k} }{ \ell_{\max,\D,k}}\right)^2  \right).
	\end{equation*}
\end{proof}

Since the maximum number of $k$-mers is $\sigma^k$, by combining the result above with the union bound we have the following result.

\begin{proposition}
	\label{prop:unionboundsample}
	Consider a sample $S$ of $m$ reads from $\D$. For fixed frequency threshold $\theta \in (0,1]$, error parameter $\varepsilon \in (0,\theta)$, and confidence parameter $\delta \in (0,1)$, if 
	\begin{equation}
		\label{eq:samplesize1_appendix}
		m \geq \frac{2}{\varepsilon^2} \left(\frac{\ell_{\max,\D,k}}{\ell_{\D,k}} \right)^2 \left( \ln \left( 2 \sigma^k\right)+\ln\left(\frac{1}{\delta} \right) \right)
	\end{equation}
	then, with probability $\ge 1- \delta$, $FK(S,k,\theta-\varepsilon/2)$ is an $\varepsilon$-approximation of $FK(\D,k,\theta)$. 
\end{proposition}
\begin{proof}
	Let $E_K$ be the event ``$|f_S(K) - f_\D(K) | \leq \frac{\varepsilon}{2}$" for a $k$-mer $K$. By the choice of $m$ and Proposition~\ref{prop:hoeffding} we have that the probability of the complementary event $\overline{E}_K$ of $E_K$ is
	\begin{equation*}
	     \Pr(\overline{E}_K) = \Pr  \left( |f_S(K) - f_\D(K) | \geq \frac{\varepsilon}{2} \right) = 2 \exp\left( -   \frac{1}{2} m \varepsilon^2 \left( \frac{\ell_{\D,k} }{ \ell_{\max,\D,k}}\right)^2 \right) \leq \frac{\delta}{\sigma^k}.
	\end{equation*}
	Now, by applying the union bound, the probability that for at least one $k$-mer $K$ of $\Sigma^k$ the event $\overline{E}_K$ holds is bounded by $\sum_{K \in \Sigma^k} \Pr(\overline{E}_K) \leq \delta$. Thus, the probability that events $E_K$ simultaneously hold for all $k$-mers $K$ in $\Sigma^k$ is at least $1-\delta$. 
	
	Now we prove that, with probability at least $1-\delta$, $FK(S,k,\theta-\varepsilon/2)$ is an  $\varepsilon$-approximation of $FK(\D,k,\theta)$, when, with probability at least $1-\delta$, ``$|f_S(K) - f_\D(K) | \leq \frac{\varepsilon}{2}$" for all $k$-mers $K$. Note that the third property of Definition~\ref{def:approximation} is already satisfied. Let $K$ be a $k$-mer of $FK(\D,k,\theta)$, that is $f_\D(K) \geq \theta$. Given that $f_S(K) \geq f_\D(K) - \varepsilon/2$, we have $f_S(K) \geq \theta - \varepsilon/2$ and the first property of Definition~\ref{def:approximation} holds. Combining $f_\D(K) \geq f_S(K) - \varepsilon/2$ and $f_S(K) \geq \theta - \varepsilon/2$, we have $f_\D(K) \geq \theta - \varepsilon$ and the second property of Definition~\ref{def:approximation} holds. 
\end{proof}

The previous theorem gives us the following simple procedure for approximating the set of frequent $k$-mers with guarantees on the quality of the solution: build a sample $S$ of $m \geq \frac{2}{\varepsilon^2} \left(\frac{\ell_{\max,\D,k}}{\ell_{\D,k}} \right)^2 \left( \ln \left( 2 \sigma^k\right)+\ln\left(\frac{1}{\delta} \right) \right)$ reads from $\D$, and output the set $FK(S,k,\theta-\varepsilon/2)$ which is an $\varepsilon$-approximation of $FK(\D,k,\theta)$ with probability at least $1-\delta$. Since the frequencies of $k$-mers we are estimating are small, then $\epsilon$ must be set to a small value. This typically results in a sample size $m$ larger than $|\D|$, making useless the sampling approach.

\section{Analysis of the First Improvement: A Pseudodimension-based Algorithm for $k$-mers Approximation by Sampling Reads}
\label{sec:advancedbounds}
In this section we prove Proposition \ref{prop:pseudobound} and Proposition  \ref{prop:samplebound2}, which here corresponds to Proposition \ref{prop:pseudobound_appendix} and Proposition \ref{prop:pseudo_reads_samplesize}, respectively. In order to help the reader to avoid too many jumps to the main text, we reintroduce some important definitions and results.

Let $\mathcal{F}$ be a class of real-valued functions from a domain $X$ to $[a,b] \subset \mathbb{R}$. Consider, for each $f \in \mathcal{F}$, the subset of $X' = X \times [a,b]$ defined as $R_f = \{ (x,t) : t \leq f(x) \}$, and call it \emph{range}. Let $\mathcal{F}^+ = \{R_f , f\in \mathcal{F}\}$ be a \emph{range set} on $X'$, and its corresponding \emph{range space} $Q'$ be $Q'=(X',\mathcal{F}^+)$. 
We say that a subset $D \subset X'$ is \emph{shattered} by $\mathcal{F}^+$ if the size of the \emph{projection set} $proj_{\mathcal{F}^+}(D) = \{ r \cap D: r \in \mathcal{F}^+ \}$ is equal to $2^{|D|}$. The \emph{VC dimension} $VC(Q')$ of $Q'$ is the maximum size of a subset of $X'$ shattered by $\mathcal{F}^+$.
The \emph{pseudodimension} $PD(X,\mathcal{F})$ is then defined as the VC dimension of $Q'$: $PD(X,\mathcal{F}) = VC(Q')$.

Let $\pi$ the uniform distribution on $X$, and let $S$ be a sample of $X$ of size  $|S| = m$, with every element of $S$ sampled independently and uniformly at random from $X$. We define, $\forall f \in \mathcal{F}$, $f_S = \frac{1}{m} \sum_{x \in S}f(x)$ and $f_X =\E_{x\sim\pi}[f(x)]$. Note that $\E[f_S] = f_X$.
The following result relates the accuracy and confidence parameters $\varepsilon$,$\delta$ and the pseudodimension with the probability that the expected values of the functions in $\mathcal{F}$ are well approximated by their averages computed from a finite random sample.

\begin{proposition}[\cite{talagrand1994sharper,long1999complexity}]
	\label{prop:pseudosamplesize_appendix}
	Let $X$ be a domain and $\mathcal{F}$ be a class of real-valued functions from $X$ to $[a,b]$. Let $PD(X,\mathcal{F}) = VC(Q') \leq v$. There exist an absolute positive constant $c$ such that, for fixed $\varepsilon,\delta \in (0,1)$, if $S$ is a random sample of $m$ samples drawn independently and uniformly at random from $X$ with
	\begin{equation}
		m \geq \frac{c \left( b-a \right)^2}{\varepsilon^2}\left(v + \ln \left(\frac{1}{\delta}\right)\right)
	\end{equation}
	then, with probability $\geq 1-\delta$, it holds simultaneously $\forall f \in \mathcal{F}$ that $| f_S  - f_X | \leq \varepsilon$.
\end{proposition}

The universal constant $c$ has been experimentally estimated to be at most $0.5$ \cite{loffler2009shape}. 


Here we define the range space associated to $k$-mers and derive an upper bound to its pseudodimension. Finally, we derive a tighter sample size compared to the one proposed in Proposition~\ref{prop:unionboundsample}. 

The definition of the range space $Q'=(X',\mathcal{F}^+)$ associated to $k$-mers requires to define the domain $X$ and the class of real-valued functions $\mathcal{F}$.

\begin{definition}
	\label{def:pseudodefs}
	Let $k$ be a positive integer and $\D$ be a bag of $n$ reads.  Define the domain $X$ as the set of integers $\{1, \dots, n\}$, where every $i \in X$ corresponds to the $i$-th read of $\D$. Then define the family of real-valued functions $\mathcal{F} = \{ f_K , \forall K \in \Sigma^k\}$ where, for every $i \in X$ and for every $f_K \in \mathcal{F}$, the function $f_K(i)$ is the number of distinct positions in read $r_i$ where $k$-mer $K$ appears divided by the average size of the multiset of $k$-mers that appear in a read of $\D$: $f_K(i) = \sum_{j=0}^{n_i-k} \frac{  \phi_{r_i,K}(j)}{\ell_{\D,k}}$. Therefore $f_K(i) \in [0 ,\frac{\ell_{\max,\D,k}}{\ell_{\D,k}} ]$.  For each $f_K \in \mathcal{F}$, the subset of $X' = X \times [0 ,\frac{\ell_{\max,\D,k}}{\ell_{\D,k}} ]$ defined as $R_{f_K} = \{ (i,t) : t \leq f_K(i) \}$ is the associated range. Let $\mathcal{F}^+ = \{R_{f_K} , f_K \in \mathcal{F}\}$ be the range set on $X'$, and its corresponding range space $Q'$ be $Q'=(X',\mathcal{F}^+)$. 
\end{definition}

A trivial upper bound to $PD(X , \mathcal{F})$ is given by $PD(X , \mathcal{F}) \leq \lfloor \log_2 |\mathcal{F}| \rfloor =\lfloor  \log_2 \sigma^k \rfloor$. Before proving a tighter bound to $PD(X , \mathcal{F})$, we first state a technical Lemma (Lemma 3.8 from~\cite{riondato2018abra}.

\begin{lemma}
	\label{lemma:lemmapseudo}
	Let $B \subseteq X'$ be a set that is shattered by $\mathcal{F}^+$. Then $B$ does not contain any element in the form $(i,0)$, for any $i \in X$.
\end{lemma}

\begin{proposition}
	\label{prop:pseudobound_appendix}
	Let $\D$ be a bag of $n$ reads, $k$ a positive integer, $X$ be the domain and $\mathcal{F}$ be the family of real-valued functions defined in Definition \ref{def:pseudodefs}. Then the pseudodimension $PD(X , \mathcal{F})$ satisfies 
	\begin{equation}
		PD(X , \mathcal{F}) \leq \lfloor \log_2(\ell_{max,\D,k}) \rfloor + 1.
	\end{equation}
\end{proposition}
\begin{proof}
	From the definition of pseudodimension we have $PD(X,\mathcal{F}) = VC(Q')$, therefore showing $VC(Q') = v \leq  \lfloor \text{log}_2(\ell_{max,\D,k}) \rfloor + 1$ is sufficient for the proof.
	An immediate consequence of Lemma \ref{lemma:lemmapseudo} is that for all elements $(i,t)$ of any set $B$ that is shattered by $\mathcal{F}^+$ it holds $t \geq 1/\ell_{\D,k}$.
	Now we denote an integer $v$ and suppose that $VC(Q') = v$. Thus, there must exist a set $B \subseteq X'$ with $|B| = v$ which needs to be shattered by $\mathcal{F}^+$. This means that $2^v$ subsets of $B$ must be in projection of $\mathcal{F}^+$ on $B$. If this is true, then every element of $B$ needs to belong to exactly $2^{v-1}$ such sets. This means that for a given $(i,t)$ of $B$, all the projections of $2^{v-1}$ elements of $\mathcal{F}^+$ contain $(i,t)$. Since $t \geq 1/\ell_{\D,k}$, there need to exist $2^{v-1}$ distinct $k$-mers appearing at least once in the read $r_i$. More formally, it needs to hold $n_i - k +1 \geq 2^{v-1}$, that implies $v \leq \lfloor \text{log}_2(n_i - k +1 ) \rfloor + 1$, $\forall (i,t) \in B$. Since $n_i - k +1 \leq \ell_{max,\D,k}$ for each $(i,t) \in B$, then $v \leq \lfloor \text{log}_2(\ell_{max,\D,k}) \rfloor + 1$, and the thesis holds.
\end{proof}

Based on the previous result, we obtain the following.

\begin{proposition}
	\label{prop:pseudo_reads_samplesize}
	Consider a sample $S$ of $m$ reads from $\D$. For fixed frequency threshold $\theta \in (0,1]$, error parameter $\varepsilon \in (0,\theta)$, and confidence parameter $\delta \in (0,1)$, if 
	\begin{equation}
		\label{eq:samplesize2_appendix}
		m \geq \frac{2}{\varepsilon^2} \left( \frac{\ell_{\max,\D,k}}{\ell_{\D,k}}\right)^2 \left( \lfloor \log_2\min( 2 \ell_{\max,\D,k} , \sigma^k ) \rfloor+\ln\left(\frac{1}{\delta} \right) \right)
	\end{equation}
	then, with probability $\ge 1- \delta$, $FK(S,k,\theta-\varepsilon/2)$ is an $\varepsilon$-approximation of $FK(\D,k,\theta)$.
\end{proposition}
\begin{proof}
	Let consider the domain $X$ and the class of real-valued functions $\mathcal{F}$ defined in Definition~\ref{def:pseudodefs}. For a given function $f \in \mathcal{F}$ (so for a given $k$-mer $K$), we have for $f_X = \E_{x\sim\pi}[f(x)]$ that
	\begin{equation*}
	 	f_X = \E_{r_i\sim\D}[f_K(i)]  =  \E_{r_i\sim\D} \left[ \sum_{j=0}^{n_i-k} \frac{  \phi_{r_i,K}(j)}{\ell_{\D,k}} \right] = \frac{1}{\ell_{\D,k}} \sum_{r_i \in \D} \frac{1}{n} \sum_{j=0}^{n_i-k}  \phi_{r_i,K}(j) = \frac{o_D(K)}{n \ell_{\D,k}} = f_\D(K), 
	\end{equation*}
	and for $f_S = \frac{1}{m} \sum_{x \in S}f(x)$ that
	\begin{equation*}
		f_S = \frac{1}{m} \sum_{r_i \in S} f_K(i) = \frac{1}{m} \sum_{r_i \in S} \sum_{j=0}^{n_i-k} \frac{  \phi_{r_i,K}(j)}{\ell_{\D,k}} = \frac{o_S(K)}{m \ell_{\D,k}} = f_S(K).
	\end{equation*}
	Combining the trivial bound $PD(X , \mathcal{F}) \leq \lfloor  \log_2 \sigma^k \rfloor$ with  Propositions~\ref{prop:pseudosamplesize} and~\ref{prop:pseudobound} we have that, with probability at least $1-\delta$, $|f_S(K) - f_\D(K)| \leq \varepsilon/2$ simultaneously holds for every $k$-mer $K$.
	
	Now, as for Proposition~\ref{prop:unionboundsample}, we prove that, with probability at least $1-\delta$, $FK(S,k,\theta-\varepsilon/2)$ is an $\varepsilon$-approximation of $FK(\D,k,\theta)$, when, with probability at least $1-\delta$, ``$|f_S(K) - f_\D(K) | \leq \frac{\varepsilon}{2}$" for all $k$-mers $K$. Note that the third property of Definition~\ref{def:approximation} is already satisfied. Let $K$ be a $k$-mer of $FK(\D,k,\theta)$, that is $f_\D(K) \geq \theta$. Given that $f_S(K) \geq f_\D(K) - \varepsilon/2$, we have $f_S(K) \geq \theta - \varepsilon/2$ and the first property of Definition~\ref{def:approximation} holds. Combining $f_\D(K) \geq f_S(K) - \varepsilon/2$ and $f_S(K) \geq \theta - \varepsilon/2$, we have $f_\D(K) \geq \theta - \varepsilon$ and the second property of Definition~\ref{def:approximation} holds. 
\end{proof}

This bound significantly improves on the result of Proposition~\ref{prop:unionboundsample}, since the factor $\ln(2 \sigma^k)$ has been reduced to $\lfloor \log_2\min( 2 \ell_{\max,\D,k} , \sigma^k ) \rfloor$. Finally, by taking a sample $S$ of size $m$ according to Proposition~\ref{prop:pseudo_reads_samplesize} and by extracting the set $FK(S,k,\theta-\varepsilon/2)$ we get an $\varepsilon$-approximation of $FK(\D,k,\theta)$ with probability at least $1-\delta$. However,  also this approach typically results in a sample size $m$ larger than $|\D|$. 

\section{Analysis of the Main Technical Result (Proposition \ref{prop:sample_bound3})}
\label{appx:sec:advancedbounds_bags}

This section is dedicated to prove our main technical result on which \algname\ is built, i.e. Proposition \ref{prop:sample_bound3} (here it corresponds to Proposition \ref{prop:pseudo_reads_samplesize_bags_appendix}). We also prove Proposition \ref{prop:pseudobound_bags_maintext} of the main text (here Proposition \ref{prop:pseudobound_bags}) and some additional but necessary results. As for the previous section, we reintroduce some important definitions and results.

We define $I_{\ell} = \{i_1 , i_2 , \dots , i_{\ell} \}$ as a \emph{bag} of $\ell$ indexes of reads of $\D$ chosen uniformly at random, with replacement, from the set $\{1,\dots,n\}$. Then we define an $\ell$-\emph{reads sample} $S_\ell$ as a bag of $m$ bags of $\ell$ reads $S_\ell = \{I_{\ell,1} , \dots , I_{\ell,m} \}$.
The definition of a new range space $Q'=(X',\mathcal{F}^+)$ associated to $k$-mers requires to define  a new domain $X$ and a new class of real-valued functions $\mathcal{F}$.

\begin{definition}
	\label{def:pseudodefs_bags}
	Let $k$ be a positive integer and $\D$ be a bag of $n$ reads.  Define the domain $X$ as the set of bags of $\ell$ indexes of reads of $\D$.  Then define the family of real-valued functions $\mathcal{F} = \{ f_{K,\ell}, \forall K \in \Sigma^k\}$ where, for every $I_\ell \in X$ and for every $f_{K,\ell} \in \mathcal{F}$, we have $f_{K,\ell}(I_\ell) = \min(1, o_{I_{\ell}}(K)) / (\ell \ell_{\D,k})$, where $o_{I_{\ell}}(K) = \sum_{i \in I_{\ell}} \sum_{j=0}^{n_i-k} \phi_{r_i,K}(j) $ counts the number of occurrences of $K$ in all the $\ell$ reads of $I_{\ell}$.
	Therefore $f_{K,\ell}(I_\ell) \in \{0 ,\frac{1}{\ell \ell_{\D,k}}\}$ $\forall f_{K,\ell}$ and $\forall I_\ell$.  
	For each $f_{K,\ell} \in \mathcal{F}$, the subset of $X' = X \times \{0 ,\frac{1}{\ell \ell_{\D,k}}\}$ defined as $R_{f_{K,\ell}} = \{ (I_\ell,t) : t \leq f_{K,\ell}(I_\ell) \}$ is the associated range. Let $\mathcal{F}^+ = \{R_{f_{K,\ell}} , f_{K,\ell} \in \mathcal{F}\}$ be the range set on $X'$, and its corresponding range space $Q'$ be $Q'=(X',\mathcal{F}^+)$. 
\end{definition}

Note that, for a given bag $I_{\ell}$, the functions $f_{K,\ell}$ are then biased if $K$ appears more than $1$ times in all the $\ell$ reads of $I_{\ell}$. 
We now prove an upper bound to the pseudodimension $PD(X , \mathcal{F})$.

\begin{proposition}
	\label{prop:pseudobound_bags}
	Let $\D$ be a bag of $n$ reads, $k$ a positive integer, $X$ be the domain and $\mathcal{F}$ be the family of real-valued functions defined in Definition \ref{def:pseudodefs_bags}. Then the pseudodimension $PD(X , \mathcal{F})$ satisfies 
	\begin{equation}
	PD(X , \mathcal{F}) \leq \lfloor \text{log}_2(\ell \ell_{max,\D,k}) \rfloor + 1.
	\end{equation}
\end{proposition}
\begin{proof}
	From the definition of pseudodimension we have $PD(X,\mathcal{F}) = VC(Q')$, therefore showing $VC(Q') = v \leq  \lfloor \text{log}_2(\ell \ell_{max,\D,k}) \rfloor + 1$ is sufficient for the proof.
	Since Lemma \ref{lemma:lemmapseudo} is also valid for the new definition of the range space $Q'=(X',\mathcal{F}^+)$, an immediate consequence is that for all elements $(i,t)$ of any set $B$ that is shattered by $\mathcal{F}^+$ it holds $t \geq 1/(\ell \ell_{\D,k})$.
	Now we denote an integer $v$ and suppose that $VC(Q') = v$. Thus, there must exist a set $B \subseteq X'$ with $|B| = v$ which needs to be shattered by $\mathcal{F}^+$. This means that $2^v$ subsets of $B$ must be in projection of $\mathcal{F}^+$ on $B$. If this is true, then every element of $B$ needs to belong to exactly $2^{v-1}$ such sets. This means that for a given $(I_\ell,t)$ of $B$, all the projections of $2^{v-1}$ elements of $\mathcal{F}^+$ contain $(I_\ell,t)$. Since $t \geq 1/(\ell \ell_{\D,k})$, there need to exist $2^{v-1}$ distinct $k$-mers appearing at least once in the bag of $\ell$ reads associated with $I_\ell$. More formally, it needs to hold $\sum_{i \in I_\ell}(n_i - k +1) \geq 2^{v-1}$, that implies $v \leq \lfloor \text{log}_2\sum_{i \in I_\ell}(n_i - k +1) \rfloor + 1$, $\forall (I_\ell,t) \in B$. Since $n_i - k +1 \leq \ell_{max,\D,k}$ for each $(I_\ell,t) \in B$ and $i \in I_\ell$, then $v \leq \lfloor \text{log}_2(\ell \ell_{max,\D,k}) \rfloor + 1$, and the thesis holds.
\end{proof}

Before showing an improved bound on the sample size, we need to define the frequency $\hat{f}_{S_\ell}(K)$ of a $k$-mer $K$ computed from the sample $S_\ell$:
\begin{equation}
	\label{eq:freq_bags_bias}
	\hat{f}_{S_\ell}(K) = \frac{1}{m}\sum_{I_{\ell,i} \in S_\ell} f_{K,\ell}(I_{\ell,i}),
\end{equation}
which is the bias version of 
\begin{equation}
	\label{eq:freq_bags_unbias}
	f_{S_\ell}(K) = \frac{1}{m}\sum_{I_{\ell,i} \in S_\ell} o_{I_{\ell}}(K) / (\ell \ell_{\D,k}).
\end{equation}
Note that $\E[f_{S_\ell}(K)] = f_\D(K)$. In order to find a relation between $\E[\hat{f}_{S_\ell}(K)]$ and $f_\D(K)$, we need the following proposition.

\begin{proposition}
	\label{prop:relation}
	Let $\tilde{f}_\D(K) = \sum_{r_i \in \D} \mathbbm{1}(K \in r_i) /n$ and $f_\D(K) = o_\D(K)/t_{\D,k}$. It holds that:
	\begin{equation}
		\frac{\ell_{\D,k}}{\ell_{\max,\D,k}} f_\D(K) \leq \tilde{f}_\D(K) \leq  \ell_{\D,k} f_\D(K).
	\end{equation}
\end{proposition}
\begin{proof}
	Let us rewrite $\tilde{f}_\D(K)$:
	\begin{equation}
		\tilde{f}_\D(K) = \sum_{r_i \in \D} \mathbbm{1}(K \in r_i) /n = \frac{\mathbbm{1}(K \in r_1)}{\ell_{\D,k}} \frac{\ell_{\D,k}}{n} + \dots +  \frac{\mathbbm{1}(K \in r_n)}{\ell_{\D,k}} \frac{\ell_{\D,k}}{n}.
	\end{equation}

	Since $\mathbbm{1}(K \in r_i) \leq o_{r_i}(K)$ for every $i \in \{1,\dots,n\}$, we have
	\begin{equation}
		\tilde{f}_\D(K) \leq  \frac{o_{r_1}(K)}{\ell_{\D,k}} \frac{\ell_{\D,k}}{n} + \dots +  \frac{o_{r_n}(K)}{\ell_{\D,k}} \frac{\ell_{\D,k}}{n} = \ell_{\D,k} f_\D(K). 
	\end{equation}

	Since $\mathbbm{1}(K \in r_i) \geq o_{r_i}(K)/\ell_{\max,\D,k}$ for every $i \in \{1,\dots,n\}$, we have
	\begin{equation}
		\tilde{f}_\D(K) \geq  \frac{o_{r_1}(K)}{n \ell_{\D,k}} \frac{\ell_{\D,k}}{\ell_{\max,\D,k}} + \dots +  \frac{o_{r_n}(K)}{n \ell_{\D,k}} \frac{\ell_{\D,k}}{\ell_{\max,\D,k}} = \frac{\ell_{\D,k}}{\ell_{\max,\D,k}} f_\D(K). 
	\end{equation}
\end{proof}

Now we show a relation between $\E[\hat{f}_{S_\ell}(K)]$ and $f_\D(K)$.
\begin{proposition}
	\label{prop:avg_bias_frequency}
	Let $\tilde{f}_\D(K) = \sum_{r_i \in \D} \mathbbm{1}(K \in r_i) /n$ and $f_\D(K) = o_\D(K)/t_{\D,k}$. Let  $S_\ell$ be a bag of $m$ bags of $\ell$ reads drawn from $\D$. Then:
	\begin{equation}
		\E[\hat{f}_{S_\ell}(K)] \geq \frac{1}{\ell \ell_{\D,k}} ( 1 - (1-\frac{\ell_{\D,k}}{\ell_{\max,\D,k}} f_\D(K))^{\ell}).
	\end{equation}
\end{proposition}
\begin{proof}
	Let us rewrite $\E[\hat{f}_{S_\ell}(K)]$:
	\begin{equation}
		\E[\hat{f}_{S_\ell}(K)] = \frac{1}{\ell \ell_{\D,k}} \E [\min(1 , o_{I_{\ell}}(K))] = \frac{1}{\ell \ell_{\D,k}}  \Pr(o_{I_{\ell}}(K) > 0).
	\end{equation}
	
	Then, we have 
	
	\begin{equation}
		\E[\hat{f}_{S_\ell}(K)] = \frac{1}{\ell \ell_{\D,k}} Pr(o_{I_{\ell}}(K) > 0) =  \frac{1}{\ell \ell_{\D,k}} ( 1 - Pr(o_{I_{\ell}}(K) = 0)) =
	\end{equation}
	
	\begin{equation}
		= \frac{1}{\ell \ell_{\D,k}} ( 1 - \prod_{i \in I_{\ell}} \Pr(o_{r_i}(K) = 0))= \frac{1}{\ell \ell_{\D,k}} ( 1 - (1-\tilde{f}_\D(K))^{\ell}),
	\end{equation}
	and since $\tilde{f}_\D(K) \geq \frac{\ell_{\D,k}}{\ell_{\max,\D,k}} f_\D(K)$ by Proposition \ref{prop:relation}, the thesis holds.
\end{proof}

Let $\theta$ be a minimum frequency threshold. Using the previous proposition, if 
\begin{equation}
	f_\D(K) \geq \frac{\ell_{\max,\D,k}}{\ell_{\D,k}}(1- (1- \ell \ell_{\D,k} \theta)^{1/\ell})
\end{equation}
with $\ell \leq 1/(\ell_{\D,k} \theta)$, then $\E[\hat{f}_{S_\ell}(K)] \geq \theta$.

\begin{proposition}
	\label{prop:pseudo_reads_samplesize_bags_appendix}
	Let $k$ and $\ell$ be two positive integers. Consider a sample $S_\ell$ of $m$ bags of $\ell$ reads from $\D$. For fixed frequency threshold $\theta \in (0,1]$, error parameter $\varepsilon \in (0,\theta)$, and confidence parameter $\delta \in (0,1)$, if 
	\begin{equation}
		\label{eq:final_sample_size}
		m \geq \frac{2}{\varepsilon^2} \left( \frac{1}{\ell\ell_{\D,k}}\right)^2 \left( \lfloor \log_2\min( 2 \ell\ell_{\max,\D,k} , \sigma^k ) \rfloor+\ln\left(\frac{1}{\delta} \right) \right) 
	\end{equation}
then, with probability at least $1 - \delta$: 
\begin{itemize}
	\item for any $k$-mer $K \in FK(\D,k,\theta)$ such that  $f_\D(A) \geq \frac{\ell_{\max,\D,k}}{\ell_{\D,k}}(1- (1- \ell \ell_{\D,k} \theta)^{1/\ell})$ it holds $\hat{f}_{S_\ell}(K)  \geq \theta - \varepsilon / 2$;
	\item for any $k$-mer $K$ with $\hat{f}_{S_\ell}(K) \geq \theta - \varepsilon/2$ it holds $f_\D(K) \geq \theta - \varepsilon$;
	\item for any $k$-mer $K \in FK(\D,k,\theta)$ it holds $f_\D(K) \geq \hat{f}_{S_\ell}(K) - \varepsilon / 2$;
	\item for any $k$-mer $K$ with $ \ell \ell_{\D,k}(\hat{f}_{S_\ell}(K) +  \varepsilon / 2) \leq 1 $ it holds $f_\D(K) \leq \frac{\ell_{\max,\D,k}}{\ell_{\D,k}}(1 - (1-\ell \ell_{\D,k} (\hat{f}_{S_\ell}(K) +  \varepsilon / 2))^{(1/\ell)})$.
\end{itemize}
\end{proposition}
\begin{proof}
	Let us consider $\hat{f}_{S_\ell}(K) = \frac{1}{m}\sum_{I_{\ell,i} \in S_\ell} f_{K,\ell}(I_{\ell,i})$ and its expectation $E[\hat{f}_{S_\ell}(K)] = \E[ f_{K,\ell}(I_{\ell,i})]$, which is taken with respect to the uniform distribution over bags of $\ell$ reads. By using Proposition \ref{prop:pseudosamplesize}, Proposition \ref{prop:pseudobound_bags}, and by the choice of $m$, we have that with probability at least  $1-\delta$ it holds $|\E[\hat{f}_{S_\ell}(K)] - \hat{f}_{S_\ell}(K)| \leq \varepsilon/2$ for every $k$-mer $K$, which implies $\hat{f}_{S_\ell}(K) \geq \E[\hat{f}_{S_\ell}(K)] - \varepsilon/2$. Using Proposition \ref{prop:avg_bias_frequency}, when $f_\D(K) \geq \frac{\ell_{\max,\D,k}}{\ell_{\D,k}}(1- (1- \ell \ell_{\D,k} \theta)^{1/\ell})$, then $\E[\hat{f}_{S_\ell}(K)] \geq \theta$ and the first part holds.
	
	By the definitions of $\hat{f}_{S_\ell}(K)$ and $f_{S_\ell}(K)$ of Equation \ref{eq:freq_bags_bias} and Equation \ref{eq:freq_bags_unbias} we have $E[\hat{f}_{S_\ell}(K)] \leq E[f_{S_\ell}(K)] = f_\D(K)$. From the proof of the first part we have $|\E[\hat{f}_{S_\ell}(K)] - \hat{f}_{S_\ell}(K)| \leq \varepsilon/2$ for every $k$-mer $K$. If we consider a $k$-mer $K$ with $f_\D(K) < \theta - \varepsilon$ we have $\hat{f}_{S_\ell}(K) \leq \E[\hat{f}_{S_\ell}(K)] + \varepsilon/2 \leq f_\D(K) + \varepsilon/2 < \theta - \varepsilon/2$ and the second part holds.
	
	Since $f_\D(K) \geq E[\hat{f}_{S_\ell}(K)]$ and $|\E[\hat{f}_{S_\ell}(K)] - \hat{f}_{S_\ell}(K)| \leq \varepsilon/2$ for every $k$-mer $K$, we have $\E[\hat{f}_{S_\ell}(K)] \geq \hat{f}_{S_\ell}(K) - \varepsilon/2$ and the third part holds.
	
	By Proposition \ref{prop:avg_bias_frequency} we have $f_\D(K) \leq \frac{\ell_{\max,\D,k}}{\ell_{\D,k}}(1 - (1-\ell \ell_{\D,k} E[\hat{f}_{S_\ell}(K)] )^{(1/\ell)})$. Using the fact that $E[\hat{f}_{S_\ell}(K)] \leq \hat{f}_{S_\ell}(K) +  \varepsilon / 2$ for every $k$-mer $K$, the last part holds.

\end{proof}

\begin{center}
	\begin{minipage}{\linewidth}
		\begin{algorithm}[H]
			\label{alg:approximation_appendix}
			\KwData{$\D$, $k$, $\theta \in (0,1]$, $\delta \in (0,1)$, $\varepsilon \in (0,\theta)$, integer $\ell \geq 1$}
			\KwResult{Approximation $A$ of $FK(\D,k,\theta)$ with probability at least $1-2\delta$}
			$m \leftarrow \lceil \frac{2}{\varepsilon^2} \left( \frac{1}{\ell\ell_{\D,k}}\right)^2 \left( \lfloor \log_2\min( 2 \ell\ell_{\max,\D,k} , \sigma^k ) \rfloor+\ln\left(\frac{2}{\delta} \right) \right) \rceil$\;
			$S \leftarrow$ sample of exactly $m \ell$ reads drawn from $\D$\;
			$T \leftarrow exact\_counting(S,k)$\;
			$A \leftarrow \emptyset$\;
			\ForAll {$k$-mers $K \in T$} 
				{
					$\hat{f}_{S_{\ell}}(K) \leftarrow Binomial(m,1-e^{-T[K]/m})/(m \ell \ell_{\D,k})$ \tcp*{biased frequency \ref{eq:freq_bags_bias}}
					$f_{S_{\ell}}(K) \leftarrow T[K]/(m \ell \ell_{\D,k})$ \tcp*{unbiased frequency \ref{eq:freq_bags_unbias}}
					\If{$\hat{f}_{S_{\ell}}(K) \geq \theta - \varepsilon/2$}{$A \leftarrow A \cup (K,f_{S_{\ell}}(K))$}
				}
			\Return $A$\;
			\caption{\algname$(\D, k, \theta, \delta, \varepsilon,\ell)$}
		\end{algorithm}
	\end{minipage}
\end{center}

\section{Additional figures}
 \begin{figure}[H]
	\centering
	\includegraphics[width=.8\linewidth]{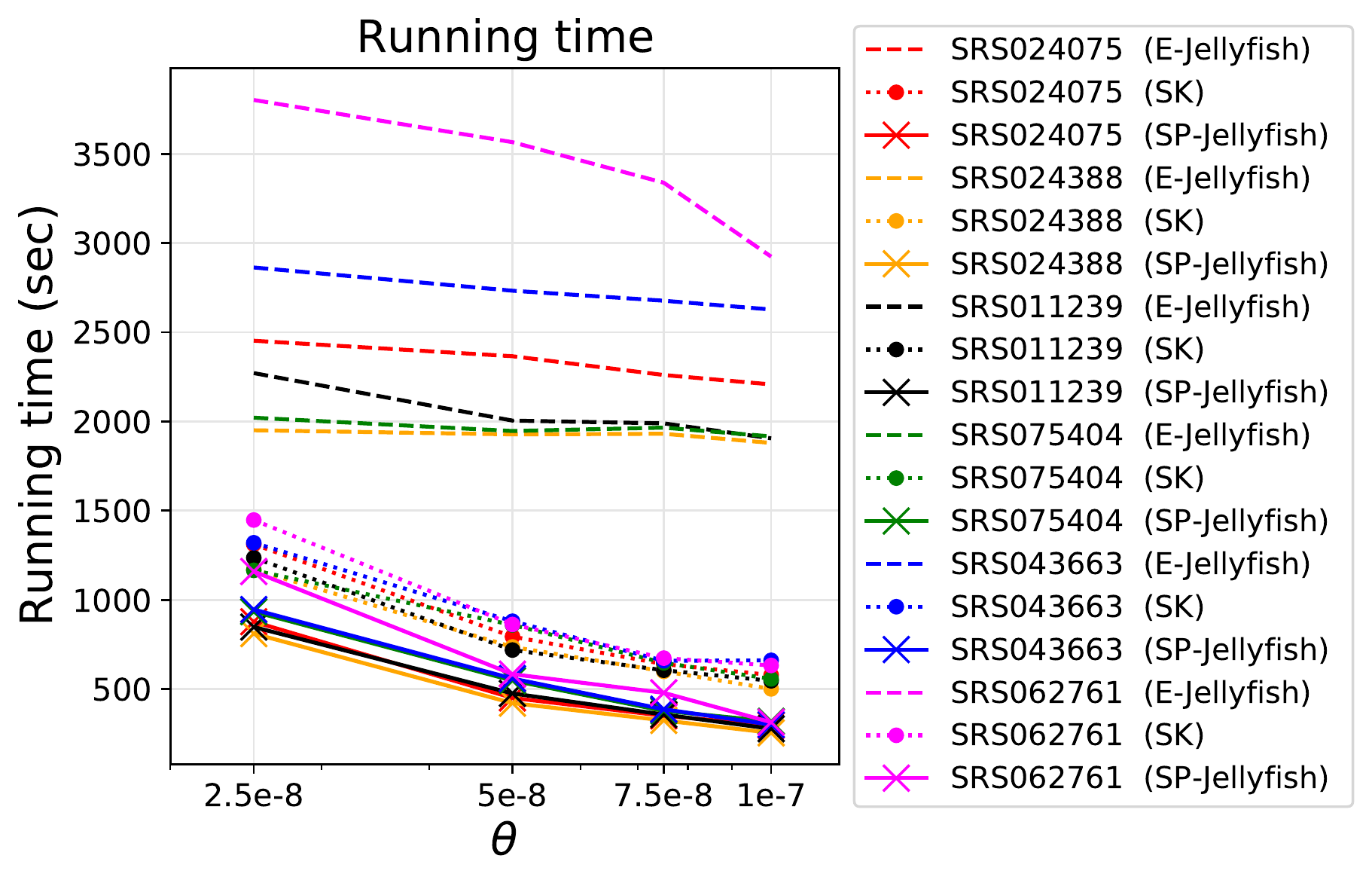}
	\caption{As function of $\theta$ and for each dataset $\D$, running times to approximate $FK(\D,k,\theta)$ with \algname\  using \jellyfish\ (\algnameshort-\jellyfish), with the state-of-the-art sampling algorithm \sakeima\ (\sakeimashort), and for exactly computing $FK(\D,k,\theta)$ with \jellyfish\ (\exactshort-\jellyfish).}
	\label{fig:runningtimes_withJelly}
\end{figure}

\begin{figure}[H]
	\centering
	\subfloat[]{\includegraphics[width=.50\linewidth]{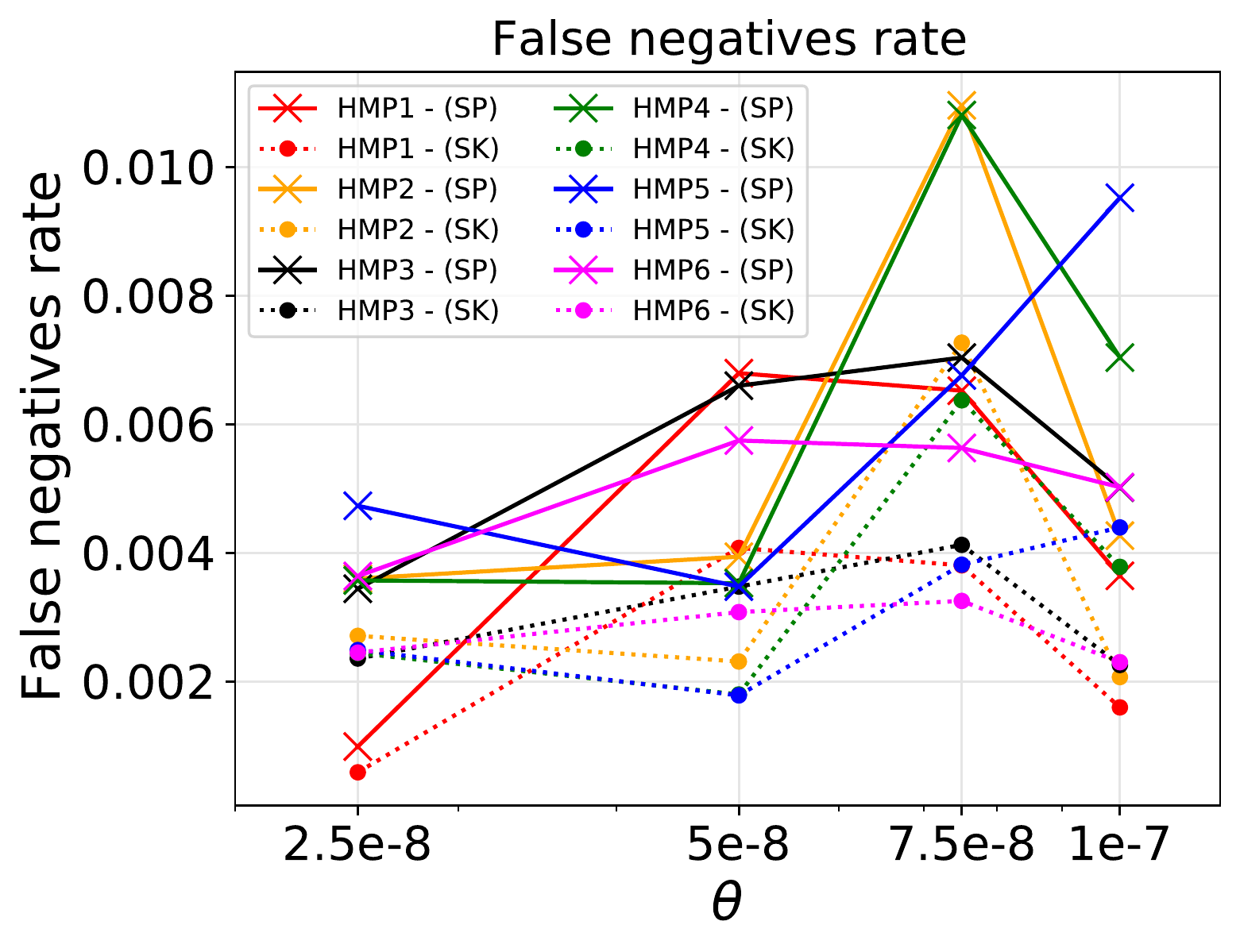} \label{fig:fn_rate}}
	\subfloat[]{\includegraphics[width=.49\linewidth]{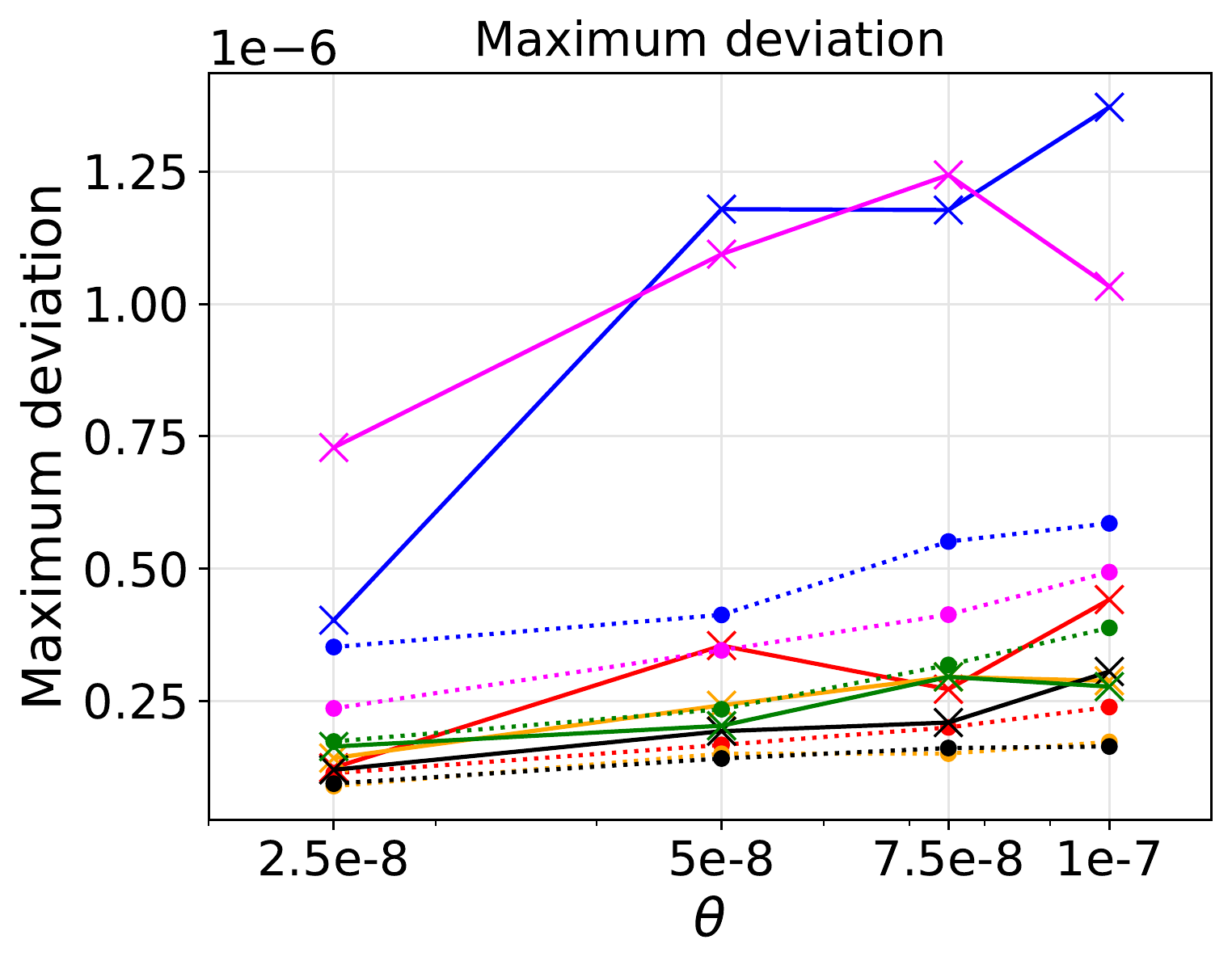} \label{fig:max_dev}}\\
	\caption{As function of $\theta$ and for every dataset: (a) False negatives rates, i.e. the fraction of $k$-mers of $FK(\D,k,\theta)$ not reported by the approximation sets, obtained using \algname\ (\algnameshort) and \sakeima\ (\sakeimashort); (b) Maximum deviations between exact and unbiased observed frequencies provided by the approximations sets of \algname\ (\algnameshort) and \sakeima\ (\sakeimashort).}
	\label{fig:quality}
\end{figure}

\begin{figure}[h]
	\centering
	\subfloat[]{\includegraphics[width=.45\linewidth]{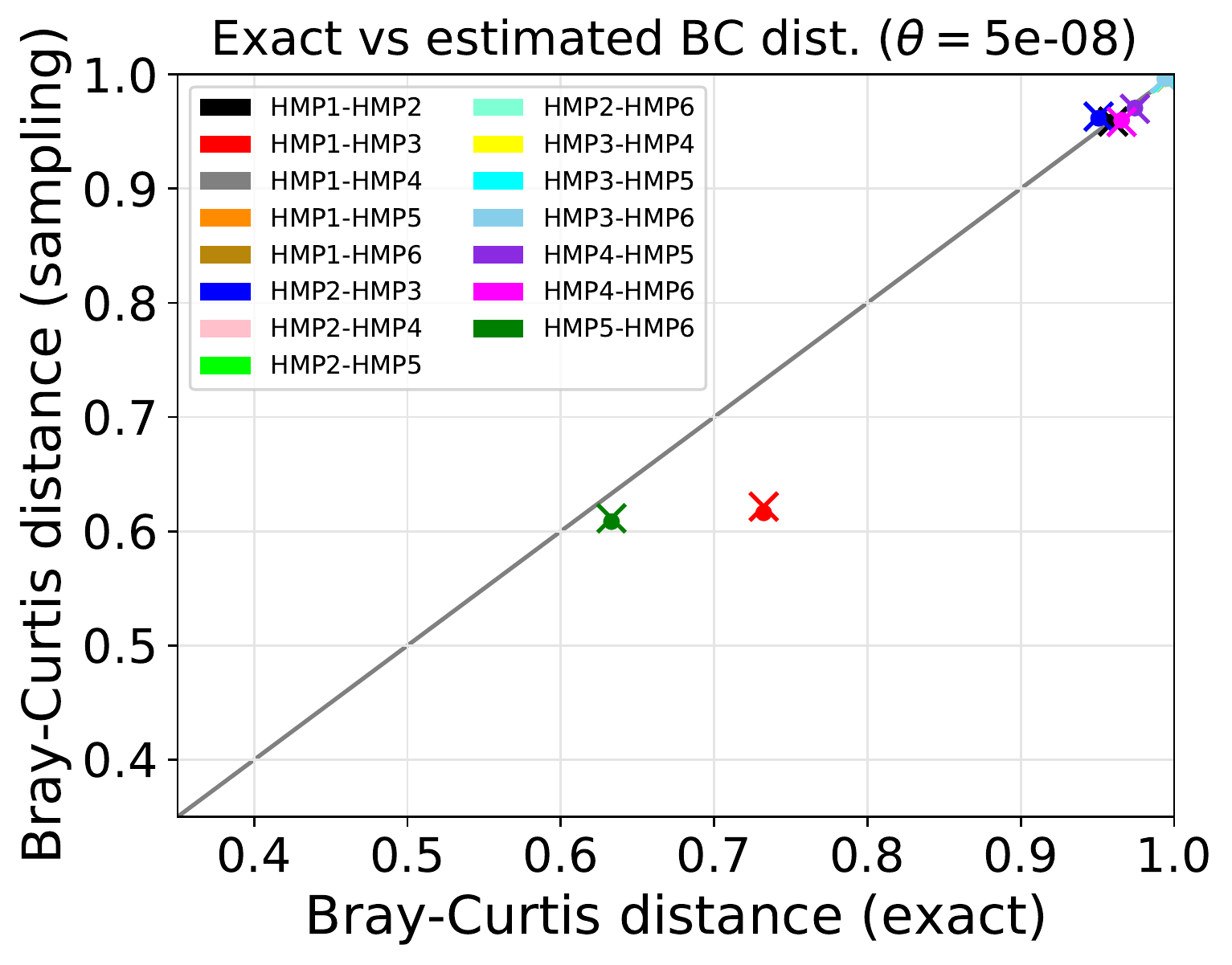} \label{fig:theta2}}
	\subfloat[]{\includegraphics[width=.45\linewidth]{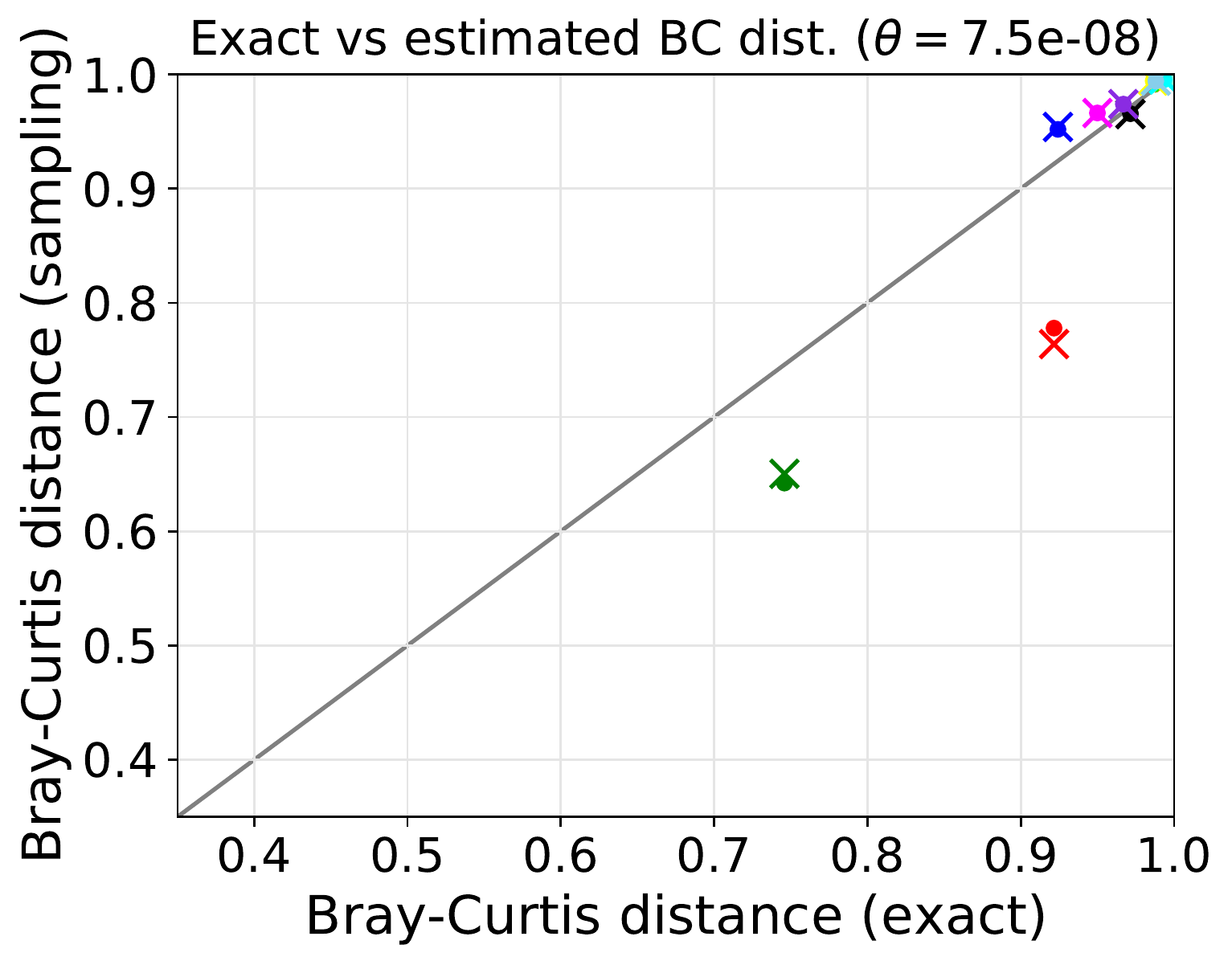}  \label{fig:theta3}}\\
	\subfloat[]{\includegraphics[width=.45\linewidth]{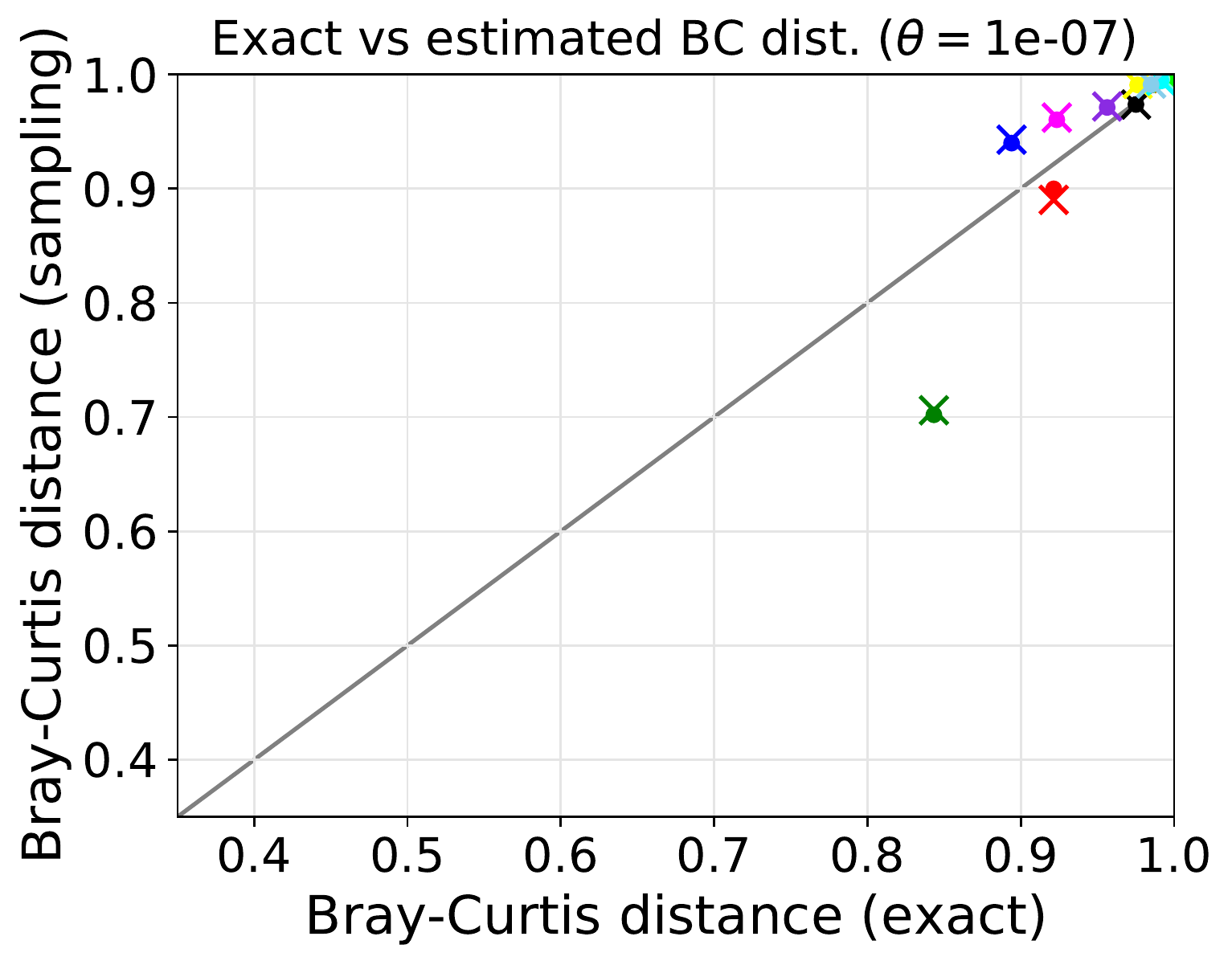}  \label{fig:theta4}}\\
	\caption{Comparison of the approximations of the Bray-Curtis distances using approximations of frequent $k$-mers sets provided by \algname\ ($\times$) and by \sakeima\ ($\bullet$) with the exact distances, for: (a) $\theta=  5\cdot10^{-8}$; (b) $\theta=  7.5\cdot10^{-8}$; (c) $\theta=  1\cdot10^{-7}$. }
	\label{fig:all_thetas}
\end{figure}

\clearpage

\begin{figure}[H]
	\centering
	\includegraphics[width=.55\linewidth]{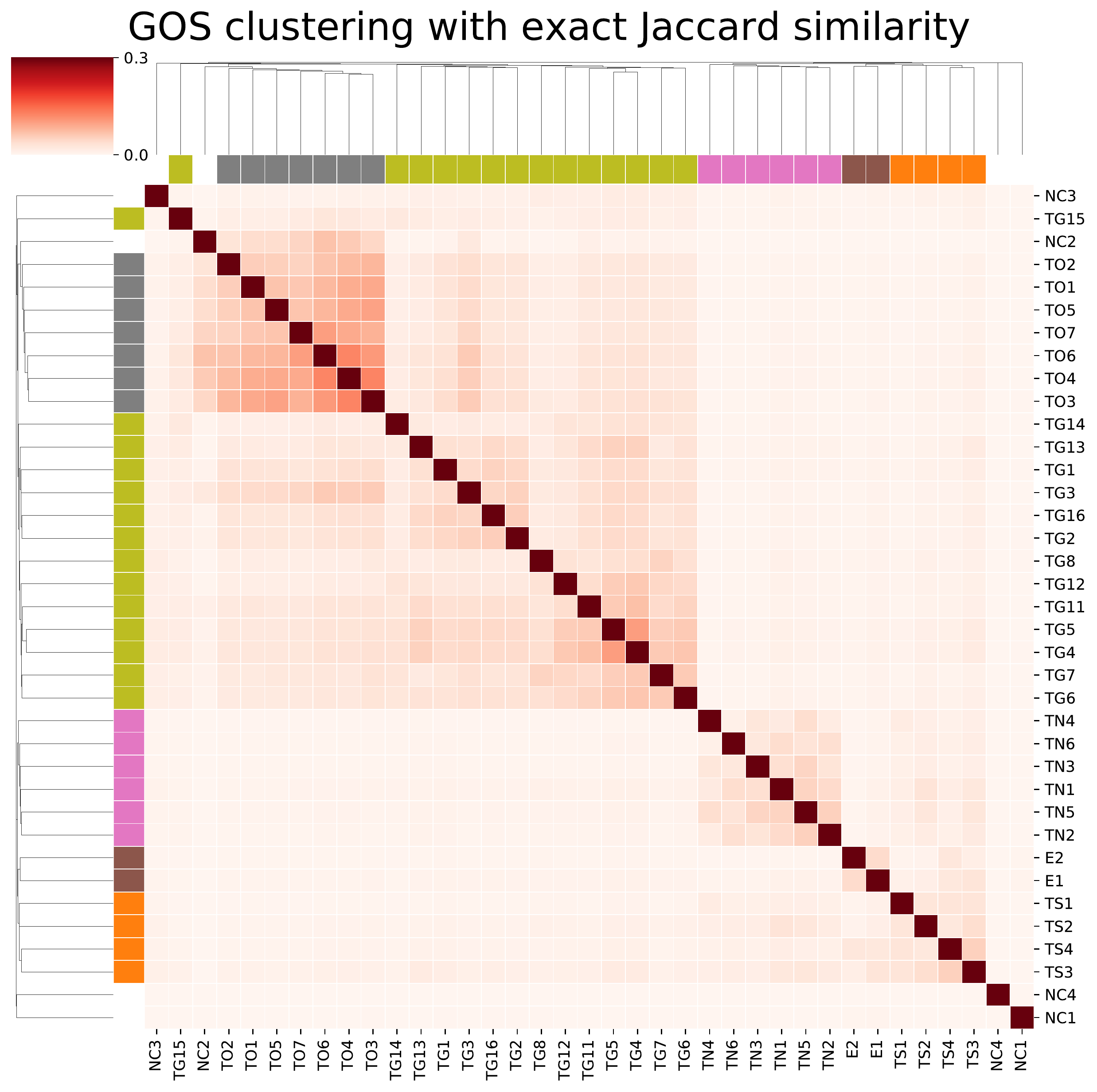}
	\caption{Average linkage hierarchical clustering of GOS datasets using Jaccard similarity. Prefix ID of the GOS datasets: TO = Tropical Open ocean, TG = Tropical Galapagos, TN = Temperate North, TS = Temperate South, E = Estuary, NC = Non-Classified as datasets from marine environments. }
	\label{fig:JC_clustering_supp}
\end{figure}
\begin{figure}[H]
	\centering
	\includegraphics[width=.55\linewidth]{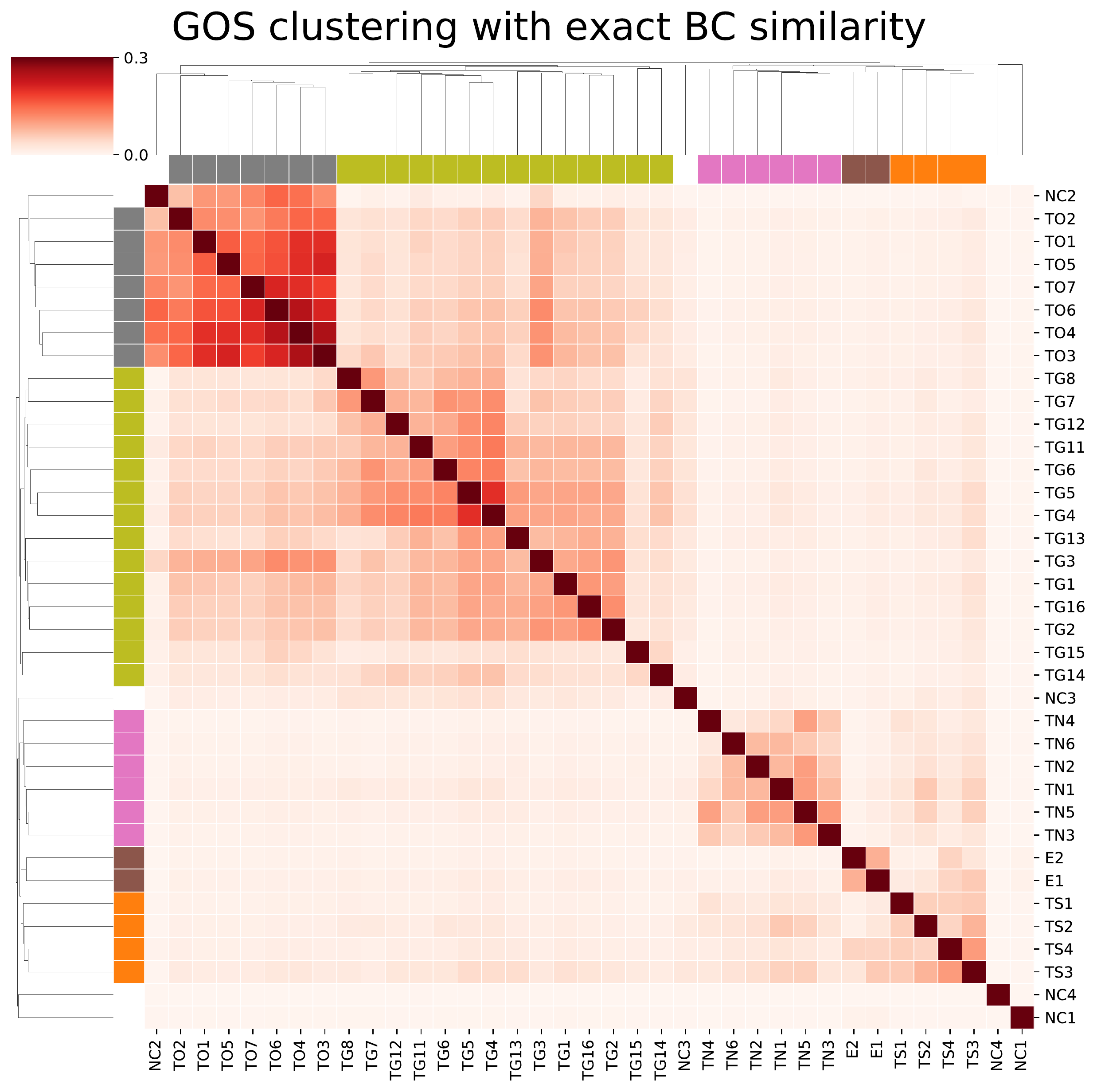}
	\caption{Average linkage hierarchical clustering of GOS datasets using Bray-Curtis similarity. Prefix ID of the GOS datasets: TO = Tropical Open ocean, TG = Tropical Galapagos, TN = Temperate North, TS = Temperate South, E = Estuary, NC = Non-Classified as datasets from marine environments. }
	\label{fig:BC_clustering}
\end{figure}
\begin{figure}[H]
	\centering
	\includegraphics[width=.55\linewidth]{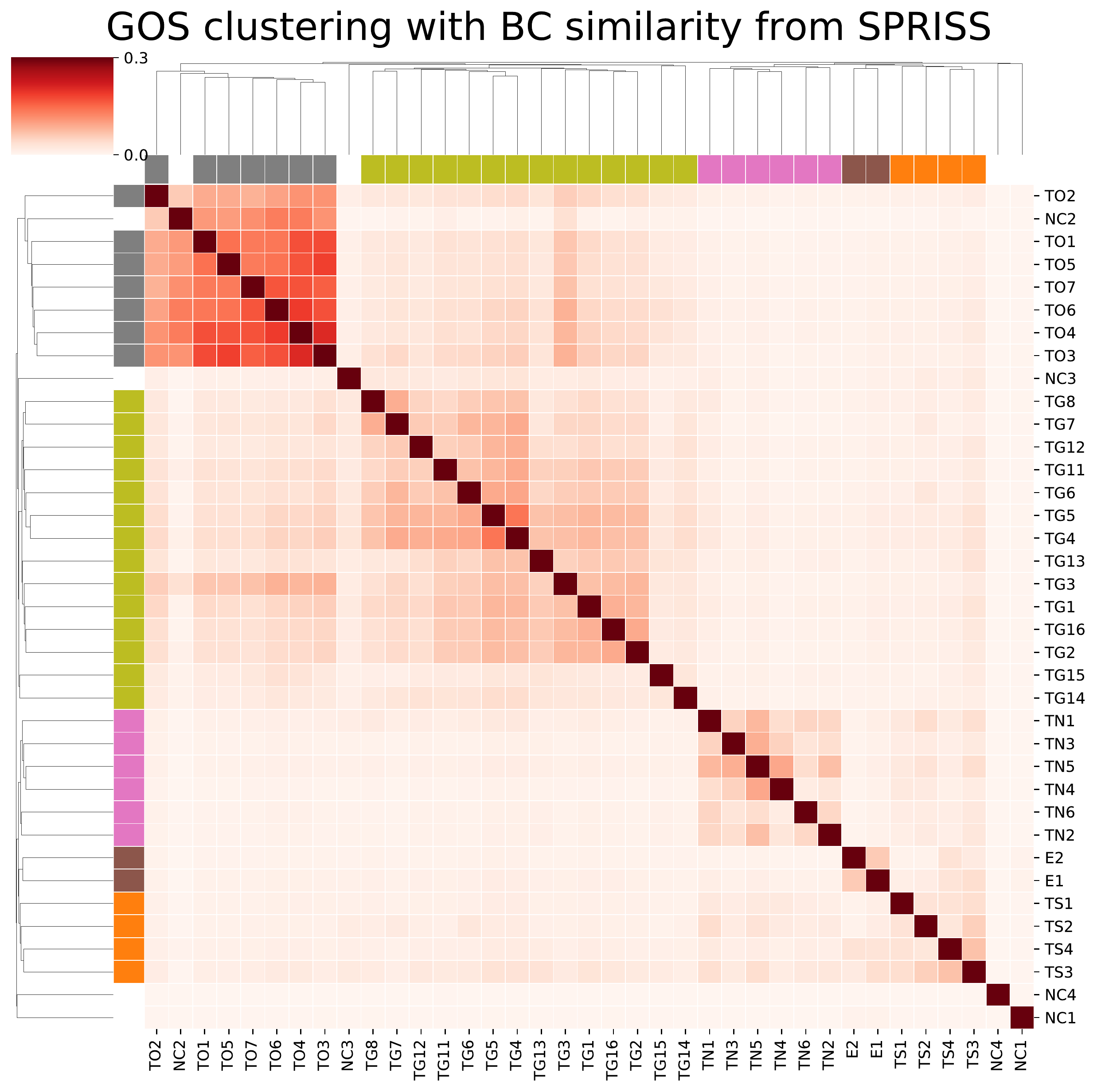}
	\caption{Average linkage hierarchical clustering of GOS datasets using estimated Bray-Curtis similarity from \algname\ with $50\%$ of the data. Prefix ID of the GOS datasets: TO = Tropical Open ocean, TG = Tropical Galapagos, TN = Temperate North, TS = Temperate South, E = Estuary, NC = Non-Classified as datasets from marine environments. }
	\label{fig:BCS_clustering_supp}
\end{figure}

\begin{figure}[h]
	\centering
	\subfloat[]{\includegraphics[width=.35\linewidth]{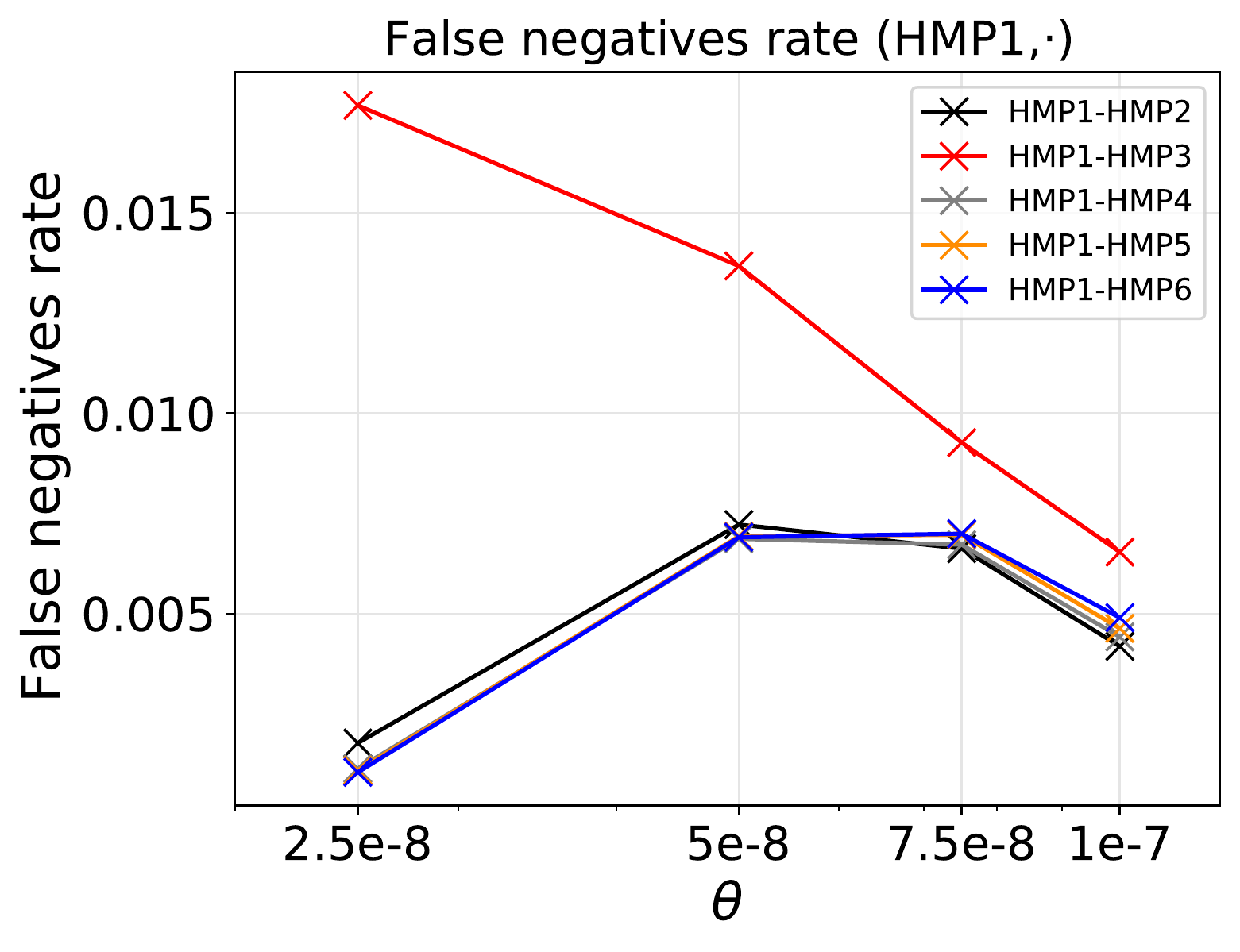}}
	\subfloat[]{\includegraphics[width=.35\linewidth]{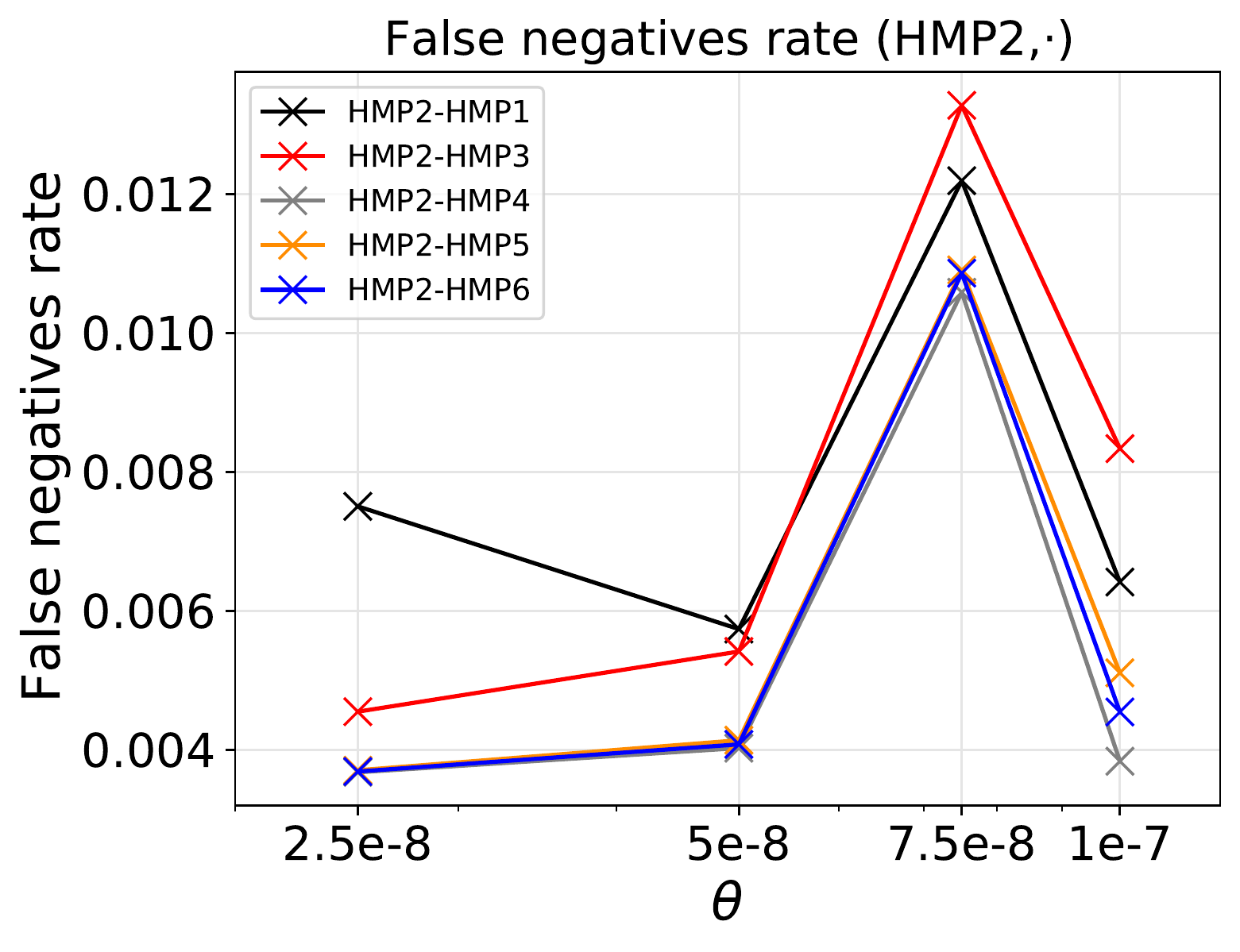}}\\
	\subfloat[]{\includegraphics[width=.35\linewidth]{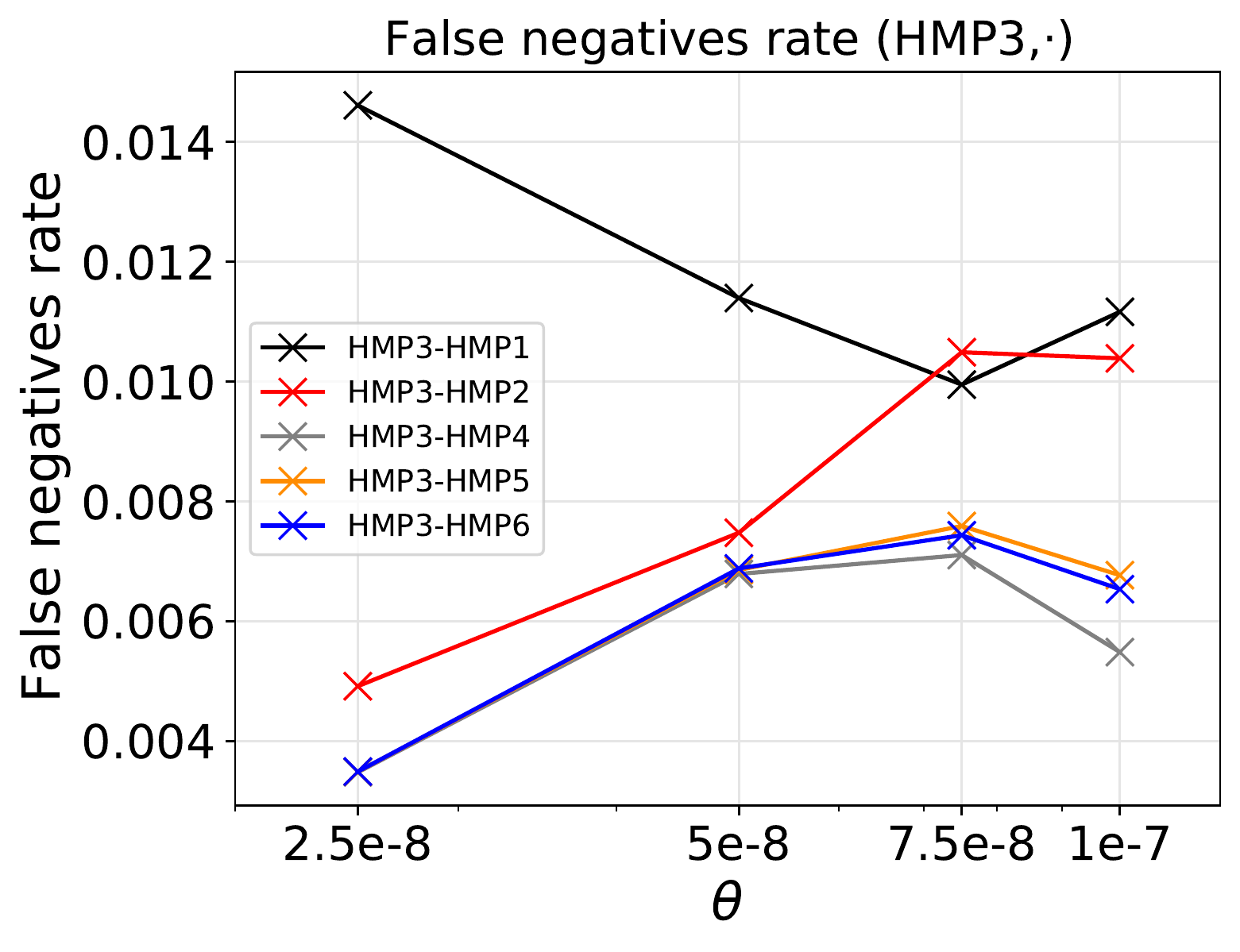}}
	\subfloat[]{\includegraphics[width=.35\linewidth]{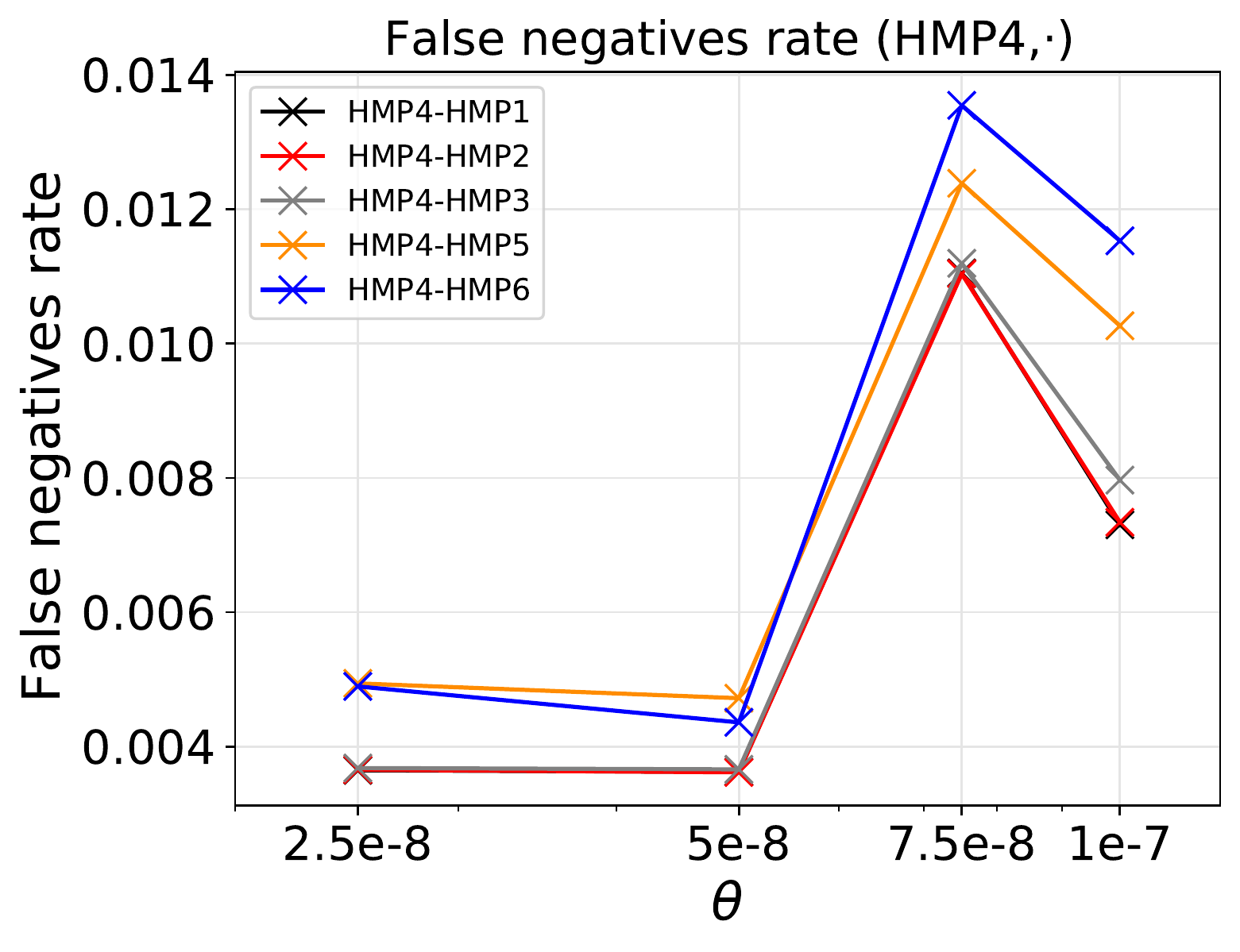}}\\
	\subfloat[]{\includegraphics[width=.35\linewidth]{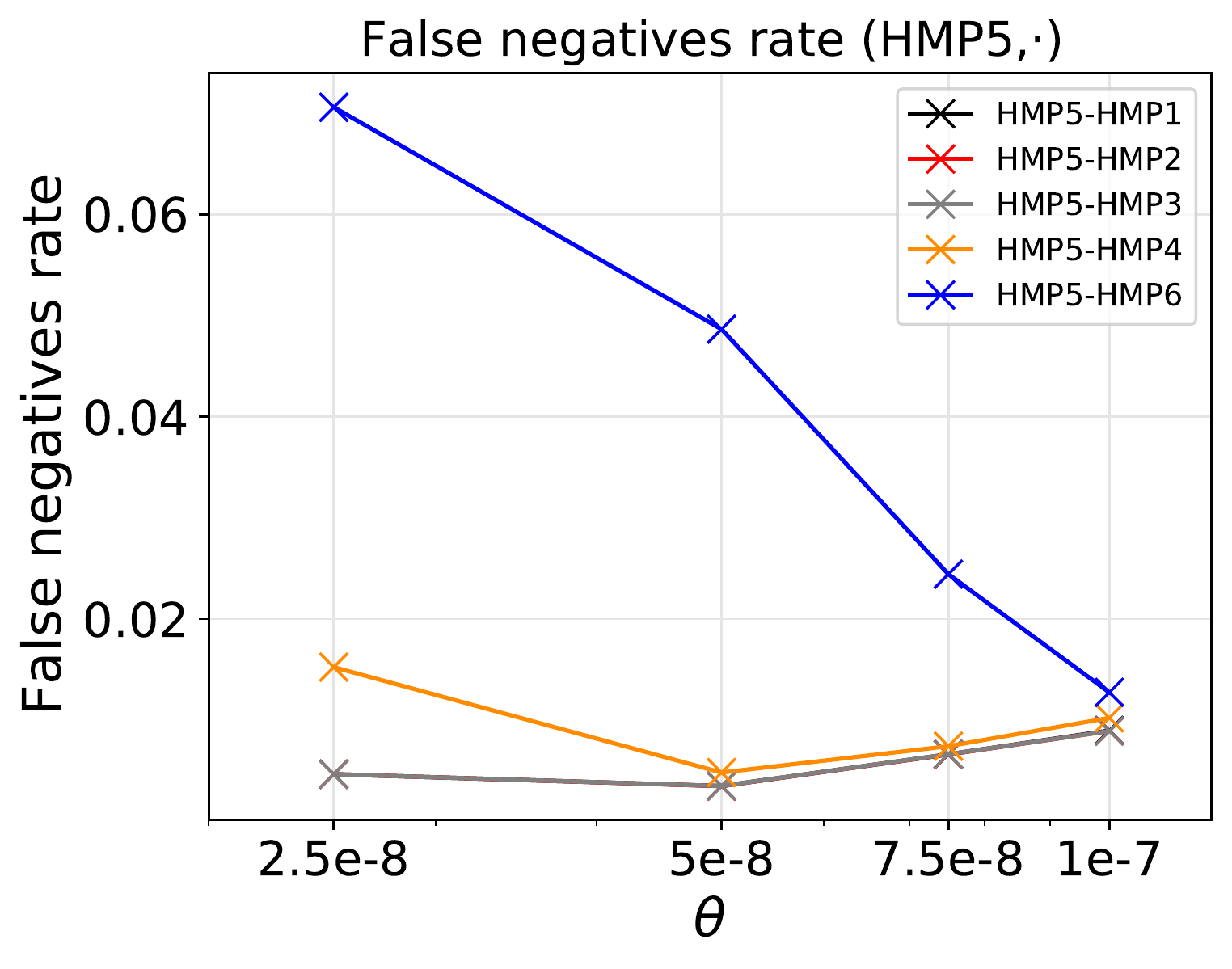}}
	\subfloat[]{\includegraphics[width=.35\linewidth]{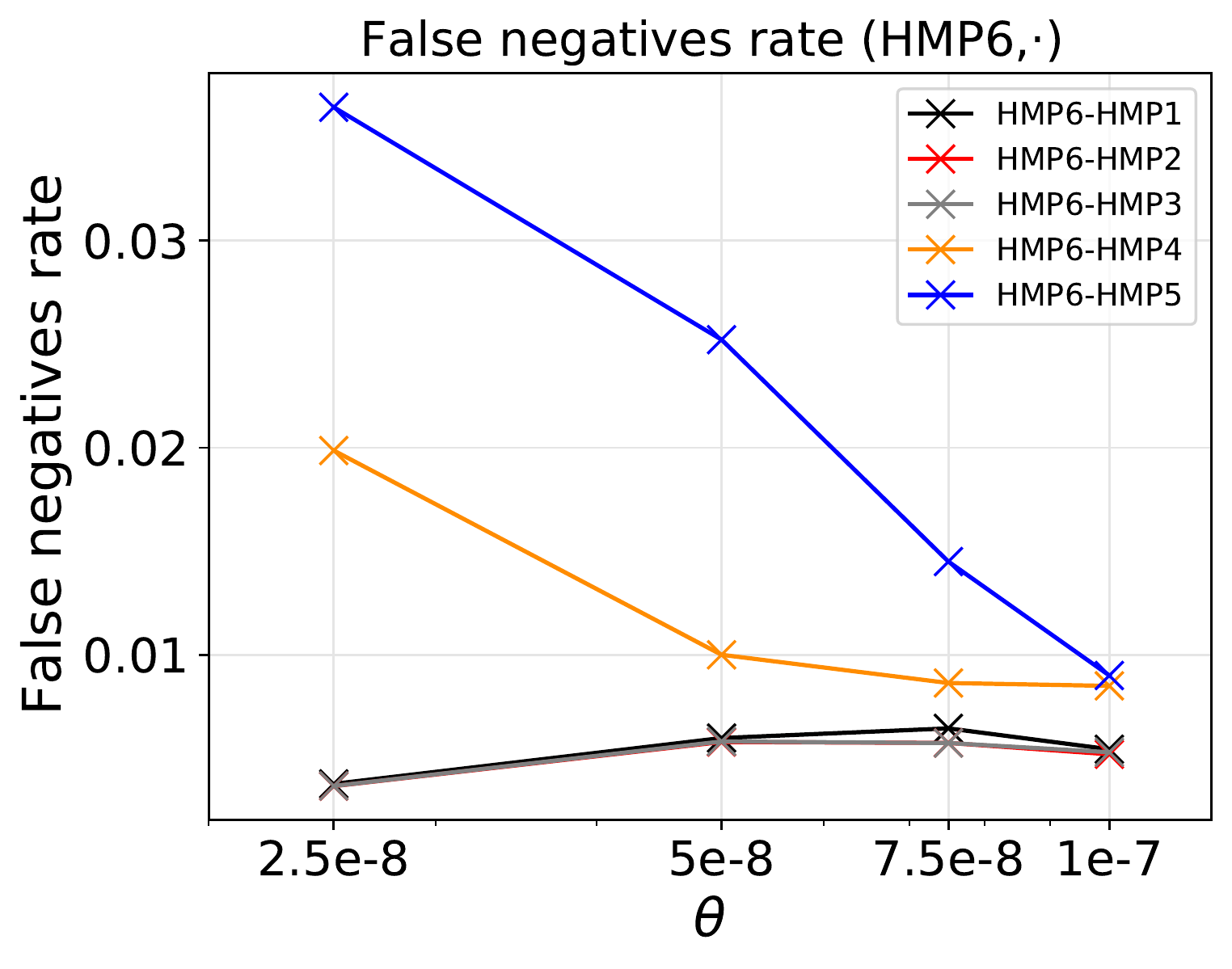}}\\
	\subfloat[]{\includegraphics[width=.35\linewidth]{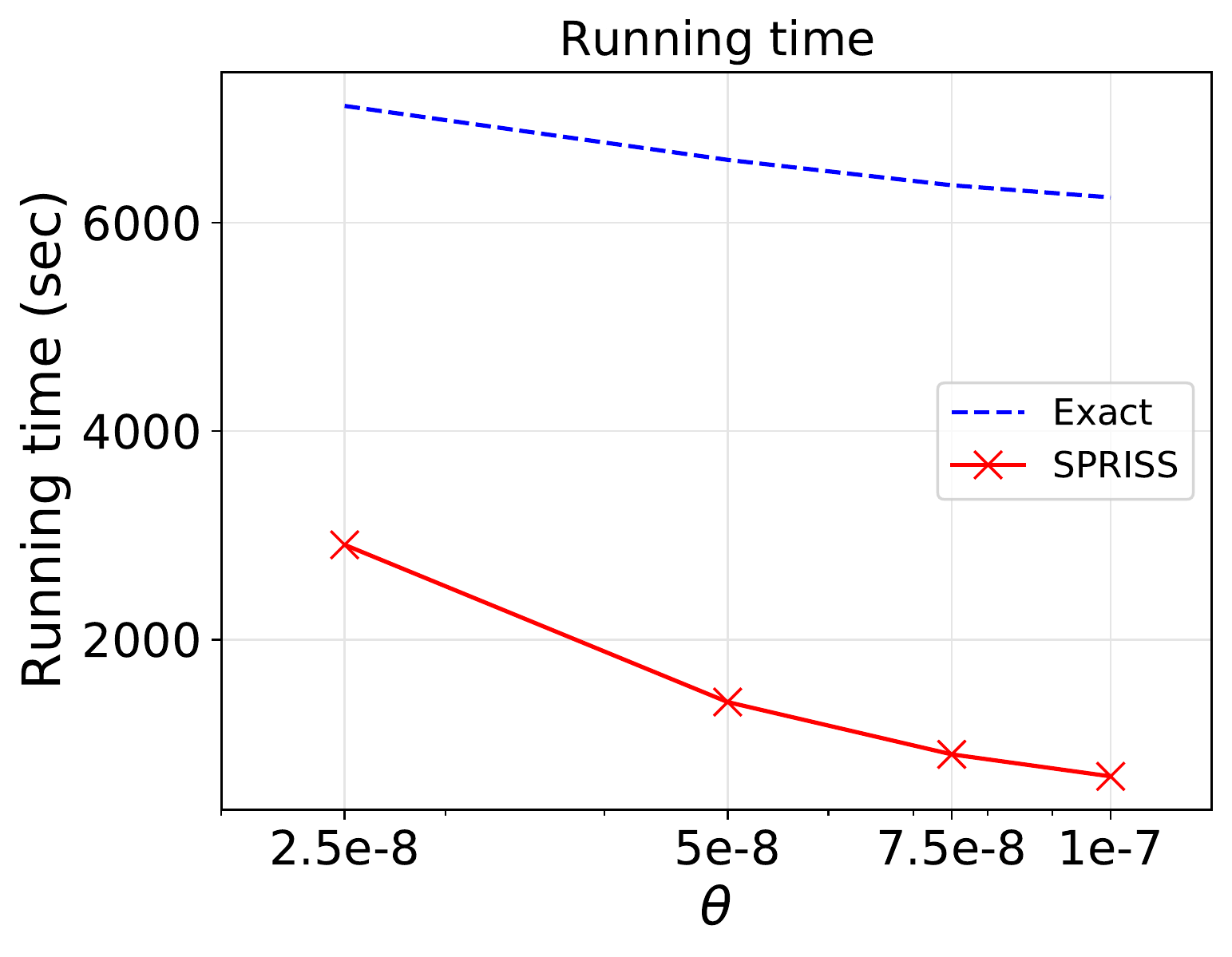} \label{fig:all_running_time_discriminative_HMP}}\\
	\caption{As function of $\theta$, false negatives rate, i.e. the fraction of $k$-mers of $DK(\D_1,\D_2,k,\theta,\rho)$  not included in its approximation $\overline{DK}(\D_1,\D_2,k,\theta,\rho)$, which is obtained using \algname, for all pairs of datasets $\D_1$ and $\D_2$. Figure \ref{fig:all_running_time_discriminative_HMP} shows the running times to compute $\overline{DK}(\D_1,\D_2,k,\theta,\rho)$ using \algname\ against the one required to compute the exact set $DK(\D_1,\D_2,k,\theta,\rho)$, cumulative for all pairs of datasets $\D_1$ and $\D_2$.} 
	\label{fig:results_discriminative_HMP}
\end{figure}

\section{Datasets}

\begin{table}[H]
\caption{HMP datasets for our experimental evaluation. For each dataset $\D$ the table shows: the dataset name and site (\texttt{(s)} for \texttt{stool}, \texttt{(t)} for \texttt{tongue dorsum}); its corresponding label on figures; the total number $t_{\D,k}$ of $k$-mers ($k=31$) in $\D$; the number $|\D|$ of reads it contains; the maximum read length $\max_{n_i} = \max_i\{n_i | r_i \in \D\}$; the average read length avg$_{n_i} = \sum_{i=1}^{n} n_i/n$.}
\label{tab:dataHMP}
   \centering
\footnotesize{
   \begin{tabular}{|c|r|r|r|r|r|r|}
      \hline
      dataset & label & \multicolumn{1}{|c|}{$t_{\D,k}$}& \multicolumn{1}{|c|}{$|\D|$} & \multicolumn{1}{|c|}{$\max_{n_i}$} & avg$_{n_i}$ \\
       \hline
       \hline
       \texttt{SRS024075(s)}  & HMP1 & $8.82 \cdot 10^9$& $1.38 \cdot 10^8$ & 95 & 93.88 \\
       \hline
       \texttt{SRS024388(s)}  & HMP2 & $7.92 \cdot 10^9$& $1.20 \cdot 10^8$ & 101 & 96.21 \\
       \hline
       \texttt{SRS011239(s)}   & HMP3 & $8.13 \cdot 10^9$& $1.24 \cdot 10^8$ & 101 & 95.69 \\
       \hline
       \texttt{SRS075404(t)}   & HMP4 & $7.75 \cdot 10^9$& $1.22 \cdot 10^8$ & 101 & 93.51 \\
       \hline
       \texttt{SRS043663(t)}   & HMP5 & $9.15 \cdot 10^9$& $1.31 \cdot 10^8$ & 100 & 100.00 \\
       \hline
       \texttt{SRS062761(t)}    & HMP6 & $8.26 \cdot 10^9$& $1.18 \cdot 10^8$ & 100 & 100.00 \\
       \hline

      \hline
   \end{tabular}
  }
\end{table}

\begin{table}[H]
	\caption{GOS datasets for our experimental evaluation. For each dataset $\D$ the table shows: the dataset name; its corresponding label for clustering results in figures \ref{fig:jaccard_clustering}, \ref{fig:bc_sampling_clustering}, and \ref{fig:BC_clustering}; the total number $t_{\D,k}$ of $k$-mers ($k=21$) in $\D$; the number $|\D|$ of reads it contains; the maximum read length $\max_{n_i} = \max_i\{n_i | r_i \in \D\}$; the average read length avg$_{n_i} = \sum_{i=1}^{n} n_i/n$. Prefix IDs of the GOS datasets: TO = Tropical Open ocean, TG = Tropical Galapagos, TN = Temperate North, TS = Temperate South, E = Estuary, NC = Non-Classified. }
	\label{tab:dataGOS}
	\centering
	\footnotesize{
		\begin{tabular}{|c|r|r|r|r|r|}
			\hline
			dataset & label & \multicolumn{1}{|c|}{$t_{\D,k}$}& \multicolumn{1}{|c|}{$|\D|$} & \multicolumn{1}{|c|}{$\max_{n_i}$} & avg$_{n_i}$ \\
			\hline
			\hline
			\texttt{GS02} & TN1  & $1.26 \cdot 10^8$& $1.21 \cdot 10^5$ & 1349 & 1058.98 \\
			\hline
			\texttt{GS03} & TN2  & $6.56 \cdot 10^7$& $6.16 \cdot 10^4$ & 1278 & 1086.07 \\
			\hline
			\texttt{GS04} & TN3  & $5.58 \cdot 10^7$& $5.29 \cdot 10^4$ & 1309 & 1074.83 \\
			\hline
			\texttt{GS05} & TN4  & $6.47 \cdot 10^7$& $6.11 \cdot 10^4$ & 1242 & 1079.37 \\
			\hline			
			\texttt{GS06} & TN5  & $6.34 \cdot 10^7$& $5.96 \cdot 10^4$ & 1260 & 1082.71 \\
			\hline
			\texttt{GS07} & TN6  & $5.44 \cdot 10^7$& $5.09 \cdot 10^4$ & 1342 & 1087.30 \\
			\hline			
			\texttt{GS08} & TS1  & $1.35 \cdot 10^8$& $1.29 \cdot 10^5$ & 1444 & 1062.24 \\
			\hline		
			\texttt{GS09} & TS2  & $8.27 \cdot 10^7$& $7.93 \cdot 10^4$ & 1342 & 1063.35 \\
			\hline
			\texttt{GS10} & TS3  & $8.08 \cdot 10^7$& $7.83 \cdot 10^4$ & 1402 & 1052.62 \\
			\hline
			\texttt{GS11} & E1  & $1.30 \cdot 10^8$& $1.24 \cdot 10^5$ & 1283 & 1070.84 \\
			\hline
			\texttt{GS12} & E2  & $1.33 \cdot 10^8$& $1.26 \cdot 10^5$ & 1349 & 1078.62 \\
			\hline
			\texttt{GS13} & TS4  & $1.46 \cdot 10^8$& $1.38 \cdot 10^5$ & 1300 & 1079.50 \\
			\hline
			\texttt{GS14} & TG1  & $1.37 \cdot 10^8$& $1.28 \cdot 10^5$ & 1353 & 1085.58 \\
			\hline
			\texttt{GS15} & TO1  & $1.35 \cdot 10^8$& $1.27 \cdot 10^5$ & 1412 & 1083.79 \\
			\hline
			\texttt{GS16} & TO2  & $1.34 \cdot 10^8$& $1.27 \cdot 10^5$ & 1328 & 1081.48 \\
			\hline
			\texttt{GS17} & TO3  & $2.76 \cdot 10^8$& $2.57 \cdot 10^5$ & 1354 & 1091.92 \\
			\hline
			\texttt{GS18} & TO4  & $1.53 \cdot 10^8$& $1.42 \cdot 10^5$ & 1309 & 1096.20 \\
			\hline
			\texttt{GS19} & TO5  & $1.43 \cdot 10^8$& $1.35 \cdot 10^5$ & 1325 & 1081.93 \\
			\hline
			\texttt{GS20} & NC1  & $3.09 \cdot 10^8$& $2.96 \cdot 10^5$ & 1325 & 1063.42 \\
			\hline
			\texttt{GS21} & TG2  & $1.40 \cdot 10^8$& $1.31 \cdot 10^5$ & 1334 & 1088.44 \\
			\hline
			\texttt{GS22} & TG3  & $1.28 \cdot 10^8$& $1.21 \cdot 10^5$ & 1288 & 1077.40 \\
			\hline
			\texttt{GS23} & TO6  & $1.40 \cdot 10^8$& $1.33 \cdot 10^5$ & 1304 & 1079.48 \\
			\hline
			\texttt{GS25} & NC2  & $1.27 \cdot 10^8$& $1.20 \cdot 10^5$ & 1288 & 1075.49 \\
			\hline
			\texttt{GS26} & TO7  & $1.06 \cdot 10^8$& $1.02 \cdot 10^5$ & 1337 & 1061.74 \\
			\hline
			\texttt{GS27} & TG4  & $2.32 \cdot 10^8$& $2.22 \cdot 10^5$ & 1259 & 1068.65 \\
			\hline
			\texttt{GS28} & TG5  & $2.01 \cdot 10^8$& $1.89 \cdot 10^5$ & 1295 & 1084.40 \\
			\hline
			\texttt{GS29} & TG6  & $1.41 \cdot 10^8$& $1.31 \cdot 10^5$ & 1356 & 1093.46 \\
			\hline
			\texttt{GS30} & TG7 & $3.84 \cdot 10^8$& $3.59 \cdot 10^5$ & 1359 & 1090.61 \\
			\hline
			\texttt{GS31} & TG8 & $4.52 \cdot 10^8$& $4.36 \cdot 10^5$ & 1341 & 1057.90 \\
			\hline
			\texttt{GS32} & NC3 & $1.50 \cdot 10^8$& $1.48 \cdot 10^5$ & 1366 & 1035.96 \\
			\hline
			\texttt{GS33} & NC4 & $7.15 \cdot 10^8$& $6.92 \cdot 10^5$ & 1361 & 1054.10 \\
			\hline
			\texttt{GS34} & TG11 & $1.39 \cdot 10^8$& $1.34 \cdot 10^5$ & 1308 & 1058.44 \\
			\hline
			\texttt{GS35} & TG12 & $1.49 \cdot 10^8$& $1.40 \cdot 10^5$ & 1321 & 1078.30 \\
			\hline
			\texttt{GS36} & TG13 & $8.42 \cdot 10^7$& $7.75 \cdot 10^4$ & 1423 & 1106.00 \\
			\hline
			\texttt{GS37} & TG14 & $6.73 \cdot 10^7$& $6.56 \cdot 10^4$ & 1244 & 1045.40 \\
			\hline
			\texttt{GS47} & TG15 & $6.70 \cdot 10^7$& $6.60 \cdot 10^4$ & 1304 & 1035.09 \\
			\hline
			\texttt{GS51} & TG16 & $1.37 \cdot 10^8$& $1.28 \cdot 10^5$ & 1349 & 1089.27 \\
			\hline
					
			\hline
		\end{tabular}
	}
\end{table}

\begin{table}[H]
	\caption{B73 and Mo17 datasets for our experimental evaluation. For each dataset $\D$ the table shows: the dataset name; the total number $t_{\D,k}$ of $k$-mers ($k=31$) in $\D$; the number $|\D|$ of reads it contains; the maximum read length $\max_{n_i} = \max_i\{n_i | r_i \in \D\}$; the average read length avg$_{n_i} = \sum_{i=1}^{n} n_i/n$.}
	\label{tab:dataMOB}
	\centering
	\footnotesize{
		\begin{tabular}{|c|r|r|r|r|r|}
			\hline
			dataset & \multicolumn{1}{|c|}{$t_{\D,k}$}& \multicolumn{1}{|c|}{$|\D|$} & \multicolumn{1}{|c|}{$\max_{n_i}$} & avg$_{n_i}$ \\
			\hline
			\hline
			\texttt{B73}   & $9.92 \cdot 10^{10}$& $4.50 \cdot 10^8$ & 250 & 250 \\
			\hline
			\texttt{Mo17}   & $9.97 \cdot 10^{10}$& $4.45 \cdot 10^8$ & 250 & 250 \\
			\hline
			
			\hline
		\end{tabular}
	}
\end{table}

\end{document}